\documentclass[11pt]{article}

\usepackage{amssymb} 
\usepackage{amsmath}     
\usepackage{amsthm}
\usepackage{todonotes}
\usepackage{mathtools}

\usepackage[ruled,vlined]{algorithm2e}
\SetKw{KwNot}{not}
\SetKw{KwOR}{or}
\SetKw{KwAND}{and}
\SetKw{KwExitFor}{exit for-loop}
\SetKw{KwStep}{Step}

\usepackage{authblk}

\usepackage{a4wide}
\usepackage{graphicx}
\usepackage{hyperref}
\newcommand{\Rset}{\mathbb{R}}
\newcommand{\Xset}{\mathbb{X}}
\newcommand{\Zset}{\mathbb{Z}}

\newcommand{\Pset}{\mathbb{P}}
\newcommand{\Uset}{\mathbb{U}}

\newtheorem{thm}{Theorem}
\newtheorem{lem}{Lemma}

\newtheorem{cor}{Corollary}
\newtheorem{prop}{Proposition}

\begin{document}

\title{Robust production planning with budgeted cumulative demand uncertainty}

\author[1]{Romain    Guillaume}
\author[2]{Adam Kasperski}
\author[2]{Pawe{\l} Zieli\'nski}

\affil[1]{Universit{\'e} de Toulouse-IRIT, Toulouse, France,
            \texttt{romain.guillaume@irit.fr}}    
\affil[2]{
Wroc{\l}aw  University of Science and Technology, Wroc{\l}aw, Poland,
            \texttt{\{adam.kasperski,pawel.zielinski\}@pwr.edu.pl}}

\date{}

\maketitle

\begin{abstract}
This paper deals with a  problem of   production planning, 
which is a version of the capacitated single-item lot sizing problem 
with backordering under demand uncertainty, modeled by uncertain 
cumulative demands.  The well-known interval budgeted uncertainty representation is assumed.
Two of its  variants  are considered. The first one is the discrete budgeted uncertainty,
in which
at most a specified  number of cumulative demands can deviate from their nominal values at the same time.
The second variant is the continuous budgeted uncertainty, in which 
the sum of the deviations of 
cumulative demands from their nominal values, at the same time, is at most 
a  bound  on the total deviation provided.
For both cases, in order to choose a production plan that hedges  against  the cumulative demand uncertainty,
 the robust minmax criterion is used.
 Polynomial algorithms for 
 evaluating the impact of uncertainty in the demand on a given production plan in terms of its cost,
 called the adversarial problem, and for
  finding  robust production plans  
  under the  discrete  budgeted uncertainty
  are constructed.
 Hence, in this case, the problems under consideration are
 not much computationally harder than their deterministic counterparts.
 For the continuous budgeted uncertainty, it is shown that the adversarial problem and the problem of computing 
 a robust production plan along with its worst-case cost  are NP-hard. In the case, when uncertainty intervals are non-overlapping, 
 they can be solved in pseudopolynomial time and admit 
  fully polynomial time
approximation schemes.  In the general case, a decomposition algorithm for
finding a robust plan is proposed.
\end{abstract}

\noindent\textbf{Keywords:} 
robustness and sensitivity analysis; production planning; demand uncertainty; robust optimization

\section{Introduction}
Production planning  under uncertainty is a fundamental and important managerial decision making
problem in various industries, among others, agriculture,  manufacturing, food and entertainment
ones (see, e.g.,~\cite{FJJXY19}).
Uncertainty  arises due to the versatility of a market and the bullwhip effect that increases  uncertainty through a supply chain (see~\cite{LPW97}). In consequence, it 
induces supply chain risks such as backordering and obsolete inventory,
accordingly, there is a need to face
 uncertainty  in planning processes, in order to manage these risks.

Nowadays most companies use the \emph{manufacturing resource planning} (MRP II)
for a manufacturing planning and control, which is
composed of three levels~\cite{ACC11}:
the  \emph{strategic level}  (production plan/resource plan), the  \emph{tactical level} 
(master production schedule/ rough-cut capacity plan and material requirement planning/capacity requirements plan) and the  \emph{operational level} (production activity control/capacity control).
 In this paper we will be concerned with the study of an impact of uncertainty 
in models that are  applicable to a production planning in the strategic level and/or a master production scheduling
 in the first level of the tactical one. More specifically, we will look more closely at 
 a  capacitated lot sizing problem with backordering under uncertainty.
  In the literature on production planning~(see, e.g.,~\cite{ADK14,DOL07,MRGL06,PMPL09,YWM98})
three different sources of uncertainty      
 such as: \emph{demand}, \emph{process} 
and \emph{supply} are distinguished. 
For the aforementioned strategic/tactical planning processes taking demand uncertainty into account plays a crucial role. Therefore, in this paper, we focus on uncertainty in the demand.
Namely, we deal with a version of the capacitated  lot sizing problem with backordering
under uncertainty in the demand.

Various models of demand uncertainty in lot sizing problems have been discussed in the literature so far,
each of which has its pros and cons.
Typically,  
 uncertainty  in demands (parameters)  is modeled by specifying a set~$\mathcal{U}$ of
all possible realizations  of the demands  (parameters), called \emph{scenarios}.
In \emph{stochastic} models,  demands
 are random variables with known probability distributions
 that induce a probability  distribution in set~$\mathcal{U}$ and
 the expected solution performance is commonly  optimized~\cite{S97,LRS07,HKMOS09,T13}.
 In  \emph{fuzzy} (\emph{possibilistic}) models,
 demands
 are modeled by
 fuzzy intervals,  regarded as \emph{possibility distributions}, describing the sets of
more or less plausible values of demands.
These fuzzy intervals induce a possibility distribution in scenario set~$\mathcal{U}$
and  some criteria  based on  \emph{possibility} measure are  optimized~\cite{GKZ12,MPP10,WL05}.

When no historical data or no information about plausible demand values
are available, required
 to
draw probability or possibility distributions
 for demands, an alternative way to handle demand uncertainty is
 a \emph{robust} model.
There are two common methods of defining scenario (uncertainty) set~$\mathcal{U}$ in this model, namely
the \emph{discrete} and \emph{interval uncertainty representations}.
 Under the discrete uncertainty representation, set~$\mathcal{U}$ is defined
by explicitly  listing all possible realizations  of  demand scenarios. While,
under the interval uncertainty representation, a closed interval is assigned to each demand,
which means that  it will take some value within this interval, but it is not possible
to predict  which one. Thus~$\mathcal{U}$ is the Cartesian product of the intervals.
In order to choose a solution, the \emph{minmax} or \emph{minmax regret} criteria
are usually  applied. As the result,
a solution minimizing its cost or opportunity loss under a worst scenario which may occur 
is computed (see, e.g.,~\cite{KY97}).
Under the discrete uncertainty representation of demands the minmax version of
a capacitated lot sizing problem turned out to be NP-hard even for  two demand scenarios
but can be solved in pseudopolynomial time, if the number of scenarios is  constant~\cite{KY97}.
When the number of scenarios is a part of the input, the is strongly NP-hard and
no  algorithm for solving it is know.
A situation is computationally much better for the interval representation of demands.
 It was shown in~\cite{GKZ12}  that
the minmax version of the lot sizing problems with backorders  can be efficiently solved.
Indeed,
the problem of computing a robust production plan with no capacity limits can be solved in $O(T)$ time,
where $T$ is the number of periods. For its capacitated  version an effective iterative  algorithm based on 
Benders' decomposition~\cite{B62}
 was provided. At  
  each iteration of the algorithm
  a worst-case 
scenario for a feasible production plan is computed 
 by a dynamic programming algorithm. Such a problem of evaluating a given production plan in terms of its
 worst-case cost is called the \emph{adversarial problem}. 
The  minmax regret version of two-stage uncapacitated lot sizing problems, studied in~\cite{Z11},
can be solved in $O(T^6\log T)$ time.

The minmax (regret) approach, commonly used in robust optimization,
 can lead to very conservative solutions.
 In~\cite{BS03, BS04} a \emph{budgeted uncertainty representation} was proposed,
 which addresses  
  this drawback under
 the interval uncertainty representation.
 It allows decision maker to flexibly control 
 the level of conservatism of  solutions computed  by specifying 
a parameter~$\Gamma$, called \emph{a budget or protection level}.
The intuition behind this was that it is unlikely that all parameters can deviate from their 
\emph{nominal} values at the same time.
The first  application of the budgeted uncertainty representation to
  lot sizing problems under demand uncertainty was proposed~\cite{BT06},
  where the resulting true minmax counterparts were approximated by
   linear programing problem (LP) and a mixed integer programing one (MIP).
  It is worth pointing out that
 the authors in~\cite{BT06} assumed a budgeted model, in which~$\Gamma$ 
 is an upper bound on the total scaled deviation of the demands (parameters) 
  from their nominal values under any demand scenario.
  Along the same line as in~\cite{BT06}, 
 the robust  optimization was adapted to periodic inventory control and production planning
 with uncertain product returns  and demand in~\cite{WYC11}, to 
 lot sizing combined with  cutting stock problems under  uncertain cost  and demand in~\cite{AM12},
 and to a production planning under make-to-stock policies in~\cite{APS18}.
 In \cite{ASNP16,AAAA17,BO08,SAP20} 
 the minmax lot sizing problems under demand uncertainty were solved  exactly 
 by  algorithms being variants of  Benders' decomposition. They consist in 
 iterative inclusion  of rows and columns, resulting from solving an  adversarial problem
  (see also~\cite{ZZ13}).
 Hence their effectiveness heavily relies on the computational  complexity of  adversarial problems.
  In~\cite{BO08}  the adversarial problems, related to computing basestock levels for
  a specific lot sizing problem,
  were solved by a dynamic programing algorithm 
  and by a MIP formulation. MIPs also model adversarial problems corresponding to
  lot sizing problems under demand uncertainty in~\cite{AAAA17,SAP20}.
  While 
  in~\cite{ASNP16}  a more general problem under parameter uncertainty containing, among others,
   a capacitated lot sizing problem under demand uncertainty,
  was examined and its  corresponding adversarial problem was solved by a dynamic programing algorithm 
  and by a fully polynomial approximation scheme (FPTAS). Therefore,
 a  pseudopolynomial algorithm and an FPTAS were proposed for a robust version of the  original general problem
 under consideration.

In most of the literature devoted to  robust 
lot sizing problems  the interval uncertainty representation is used to model uncertainty  in demands
(see, e.g., \cite{APS18,AM12,ASNP16,AAAA17,BO08,GKZ12,SAP20,WYC11,Z11}).
This is
not surprising, as it is one of the simplest and most natural ways of  handling the uncertainty.
In majority of papers the demand uncertainty   is interpreted as the uncertainty in actual demands in periods.
In this case, however, the uncertainty
cumulates in the subsequent periods due to the interval addition, which may be unrealistic in applications. In~\cite{GTZ17}  the uncertainty in cumulative demands was modeled, which resolves this problem. Indeed, it is able to take uncertain
demands in periods  as well as dependencies between the periods into account~\cite{GTG10}.

In this paper a capacitated  lot sizing problem 
with backordering under the cumulative demand uncertainty is discussed. The uncertainty in cumulative demands is modeled by using 
two variants of the interval budgeted uncertainty. In the first one, called the \emph{discrete budgeted uncertainty}~\cite{BS03, BS04},
at most a specified  number of cumulative demands~$\Gamma^d\in \Zset_{+}$ can deviate from their nominal values at the same time.
In the second variant, called the \emph{continuous budgeted uncertainty}~\cite{NO13}, 
the total  deviation of 
cumulative demands from their nominal values, at the same time, is at most~$\Gamma^c\in \Rset_{+}$.
The latter variant is similar to that considered in~\cite{BT06}, where the total scaled deviation is upper bounded
by~$\Gamma\in \Rset_{+}$, but it is different from the computational point of view.
To the best of our knowledge, the above lot sizing problem with  uncertain cumulative demands under
the discrete and continuous budgeted uncertainty has not been investigated in the literature so far.

\noindent\emph{Our contribution:} The purpose of this paper is not to motivate the robust 
approach in lot sizing problems (this was well done in other papers), but rather to provide a complexity characterization, containing
both positive and negative results for the problem under consideration.
For both variants of the budgeted uncertainty, we analyze the adversarial problem, denoted by  \textsc{Adv},
and the problem of finding a minmax (robust) production plan
along with its worst-case cost, denoted by \textsc{MinMax}.
We first consider the restrictive case, in which  the cumulative demand intervals are non-overlapping (it is actually possible 
for the master production scheduling) and we study then the general case in which the intervals can overlap.
Under the discrete budgeted uncertainty, we  provide polynomial algorithms for 
the \textsc{Adv} problem  and polynomial linear programming based methods for  the \textsc{MinMax}  problem
in the non-overlapping case. We then extend these results to the general case by showing a characterization of optimal  
 cumulative demand scenarios for~\textsc{Adv}. In consequence, 
under the discrete budgeted model
the \textsc{Adv} and \textsc{MinMax} problems can be solved efficiently.
Under the continuous budgeted uncertainty the situation is different. In particular, we prove that the \textsc{Adv} problem and, in consequence, the \textsc{MinMax}  one
are   NP-hard even in the non-overlapping case.
For the non-overlapping case, we construct  pseudopolynomial algorithms for the \textsc{Adv} problem
and propose a pseudopolynomial ellipsoidal algorithm and a linear programming program
with  pseudopolynomial number of constraints and variables for the \textsc{MinMax}  problem.
Moreover, we show that both problems admit an FPTAS.
In the general case, the \textsc{Adv} and \textsc{MinMax} problems still remain  NP-hard. Unfortunately, in this case
there is no easy characterization of vertex cumulative demand scenarios. Accordingly, 
 we  propose a MIP based approach for \textsc{Adv}
 and an exact solution algorithm, being a variant of  Benders' decomposition,  for \textsc{MinMax}.
 
 This paper is organized as follows.
 In Section~\ref{sPre} we formulate
 a deterministic capacitated  lot sizing problem with backordering and
 present a model of the cumulative demand uncertainty, i.e.
 the  discrete and continuous budgeted uncertainty representations,
 together with the \textsc{Adv} and \textsc{MinMax} problems corresponding to the  lot sizing problem.
 In Sections~\ref{sdbu} and~\ref{scbu} we study the problem under both budgeted uncertainty representations
 providing positive and negative results.
 We finish the paper with some conclusions in Section~\ref{con}.

\section{Problem formulation}
\label{sPre}
In this section we first recall a formulation of the  deterministic capacitated single-item lot sizing problem with backordering. Then we assume that demands are subject to uncertainty
and  present a model of uncertainty, called the \emph{budgeted uncertainty representation}.
In order to choose a robust production plan we apply the minmax criterion, commonly used in robust optimization.

\subsection{Deterministic production planning problem}

We are given $T$ periods, a \emph{demand}~$d_t\geq 0$ in each  period~$t$, $t\in[T]$ ($[T]$ denotes the set $\{1,\ldots,T\}$),
\emph{production}, \emph{inventory} and \emph{backordering costs}  and 
 a \emph{selling price},  denoted by  $c^P$, $c^I$, $c^B$ and~$b^P$, respectively, which do not depend on period~$t$.
Let  $x_{t}\geq 0$ be a \emph{production amount} in period~$t$, $t\in [T]$.
We assume that  the production amounts  $x_1,\ldots,x_T$, called the \emph{production plan}, can be under some linear constraints.
Namely,
let $\Xset\subseteq \Rset^T_{+}$ be a set of production plans, specified by  linear constraints, for instance:
\[
\Xset=\{\pmb{x}=(x_1,\ldots,x_{T}):\; x_t\geq 0,  l_{t}\leq x_{t}\leq u_{t}, 
                                              L_{t}\leq \sum_{i\in [t]}x_{i}\leq U_{t},
                                              t\in [T]\}\subseteq \Rset^T_{+},                                      
\]
where $l_{t}$, $u_{t}$ and $L_{t}$, $U_{t}$
are prescribed  \emph{capacity} and \emph{cumulative capacity limits}, respectively.
Accordingly, we wish to   find a feasible production plan  
$\pmb{x}\in \Xset$, subject to the conditions of satisfying each demand and the capacity limits, which  minimizes the total
production, storage and backordering costs minus the benefit from  selling the product. 

Set $D_{t}=\sum_{i\in [t]}d_{i}$ and
 $X_{t}=\sum_{i\in [t]}x_{i}$, i.e.
 $D_{t}$ and  $X_{t}$ stand for
 the \emph{cumulative demand} up to period~$t$
and the \emph{cumulative production} up to period~$t$, respectively.
We do not examine additional production processes, for example with setup times and costs, 
which lead to NP-hard problems even for special cases 
(see, e.g.,~\cite{FLK80,CT90}).
The problem under consideration is a version of 
the \emph{capacitated single-item lot sizing problem with backordering} (see, e.g.,~\cite{ABDN17,PW06}).
It
can be represented by the following linear program:
\begin{align}
 \min &\sum_{t\in [T]}(c^{I} I_{t}+ c^B B_{t}+c^Px_t-b^P s_t)&\label{spp1}\\ 
\text{ s.t. }& B_{t}- I_{t}=D_{t}-X_t & t \in [T],\label{spp2}\\
 &\sum_{i\in [t]} s_{i}=D_{t}-B_{t} & t \in [T],\label{spp3}\\
 &X_t=\sum_{i\in [t]} x_i& t \in [T],\label{spp3a}\\
        &B_{t},I_{t},s_t\geq 0  & t\in [T],\label{spp4}\\
        &\pmb{x}\in \Xset\subseteq \Rset^T_{+}, \label{spp5}
\end{align}
where $I_{t}$, $B_{t}$ and $s_{t}$ are
\emph{inventory level}, \emph{backordering  level} and  \emph{sales
 of the product}  at the end of   period $t\in [T]$, respectively.
 We assume that the initial inventory and  backorder levels are equal to~0.
There is an optimal solution to (\ref{spp1})-(\ref{spp5}) which satisfies $B_tI_t=0$ for each $t\in [T]$, so inventory storage  from period~$t$ to period~$t+1$
and backordering  from period~$t+1$ to
period~$t$ are not performed simultaneously. Indeed, if $B_t>0$ and $I_t>0$, then we can modify the solution so that $B_t=0$ or $I_t=0$ without violating the constraints~(\ref{spp2})-(\ref{spp3}) and increasing the objective value.
Using this observation, we can rewrite (\ref{spp1})-(\ref{spp5}) in the following equivalent compact form,
which is more convenient to analyze:
\begin{equation}
\min_{\pmb{x}\in \Xset} \left(\sum_{t\in [T]}\max\{ c^{I}(X_t-D_t), c^{B}(D_t-X_t)\}+c^PX_T
     -b^P\min\{ X_T,D_T\}\right).
\label{ddls}
\end{equation}
Define
$$
f_I(X_t,D_t)=\left\{\begin{array}{lll}
				c^I(X_t-D_t) &\text{if } t\in [T-1],\\
				c^I(X_t-D_t)+c^PX_T-b^PD_t &\text{if } t=T,
				\end{array}\right.
				$$

$$
f_B(X_t,D_t)=\left\{\begin{array}{lll}
				c^B(D_t-X_t) &\text{if }  t\in [T-1],\\
				c^B(D_t-X_t)+c^PX_T-b^PX_t &\text{if } t=T.
				\end{array}\right.
				$$
Hence, after considering two cases, namely $X_t\leq D_t$ and $X_t>D_t$, for each $t\in [T]$,  problem~(\ref{ddls})
 can be represented as follows
\begin{equation}
\min_{\pmb{x}\in \Xset} \sum_{t\in [T]}\max\{f_I(X_t,D_t), f_B(X_t,D_t) \}.
\label{ddls1}
\end{equation}
From now on, we will refer to~(\ref{ddls1}) instead of (\ref{spp1})-(\ref{spp5}). Observe that $f_I(X_t,D_t)$ is nonincreasing function of $D_t$, $f_B(X_t,D_t)$ is nondecreasing function of $D_t$.
Furthermore, the function $\max\{f_I(X_t,D_t), f_B(X_t,D_t)\}$  is convex in $X_t$ and $D_t$, since
 both functions~$f_I$ and~$f_B$ are linear in $X_t$ and $D_t$.

\subsection{Robust production planning problem}

We now admit that  demands in problem~(\ref{ddls1}) are subject to uncertainty. 
In practice, 
a knowledge about uncertainty in a demand is  often expressed as
$\pm\Delta$, where $\Delta$ is a possible deviation from a nominal demand value. This 
means that the actual demand
will take some value within the interval determined by~$\pm\Delta$, but it is not possible  to
predict at present which one. In consequence, a simple and
 natural \emph{interval uncertainty representation} is induced.

A demand has a twofold interpretation, namely an actual demand $d_t$ in period~$t$ or a cumulative demand $D_t$
up to period~$t$, $t\in[T]$.
The former interpretation is often considered in the literature.
However, in this case the deviations
$\Delta$ cumulate  in subsequent periods, due to the interval addition,
$\widetilde{D}_t=\widetilde{\sum}_{t\in[T]}\tilde{d}_t$, $d_t\in \tilde{d}_t= [\hat{d}_t-\Delta,\hat{d}_t+\Delta]$,
where $\hat{d}_t $ is the nominal value of the demand in period~$t$. This may be unrealistic in practice, because the deviation for cumulative demand in period $t$ becomes $t\Delta$  (see Figure~\ref{fdemands}a).
Therefore, in this paper we use a model of uncertainty in  the cumulative demands rather than in actual demands, 
expressed by  symmetric intervals~$\widetilde{D}_t= [\widehat{D}_t-\Delta, \widehat{D}_t+\Delta]$, $t\in[T]$, prescribed,
where $\widehat{D}_t$ is the nominal  value of the cumulative demand up to period $t\in [T]$ (see Figure~\ref{fdemands}b).
Such a model  is able  to take the uncertain demands in periods  into account
($\tilde{d}_t$ lies in $\widetilde{D}_t$)
 as well as dependencies between periods. We allow different deviations for each period.
Accordingly, the value of cumulative demand~$\widetilde{D}_t$ is only known to belong to the interval
$[\widehat{D}_t-\Delta_t, \widehat{D}_t+\Delta_t]$, $t\in [T]$, 
where $ 0\leq\Delta_t\leq \widehat{D}_t$ is the maximum deviation of the cumulative demand from its nominal value~$\widehat{D}_t\geq 0$. 
Each feasible vector $\pmb{D}=(D_1,\ldots,D_T)$ of cumulative demands, called \emph{scenario},
 must satisfy the following two reasonable
 constraints: $D_t\geq 0$ for every $t\in [T]$ and $D_t\leq D_{t+1}$ for every  $t\in [T-1]$. Since the vector of  nominal cumulative demands should be feasible, we also assume that $\widehat{D}_t\leq \widehat{D}_{t+1}$, $t\in [T-1]$.
Notice that scenario $(D_1,\ldots,D_T)$
  induces  a vector of actual demands  in periods~$t$, i.e. 
  $d_1=D_1$,
   $d_t=D_t-D_{t-1}$, $t=2,\ldots,T$.
 \begin{figure}[ht]
 \centering
 \includegraphics[scale=0.35]{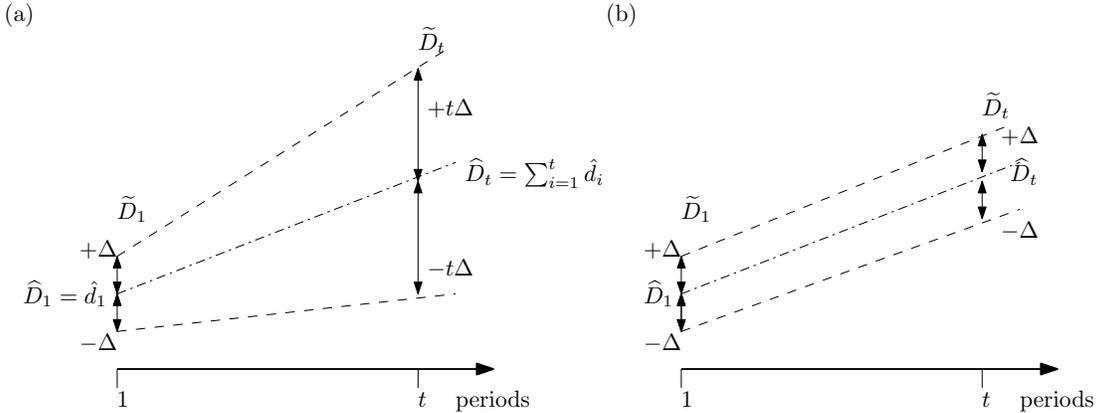}
  \caption{(a)  Demands in  periods under the interval uncertainty. 
                (b)  Cumulative demands under the interval uncertainty. } \label{fdemands}
 \end{figure}

 In this paper we study the following two special cases
of the interval, symmetric uncertainty representations, that are common  in robust optimization,
in which \emph{scenario sets} are defined in the following way
(see, e.g.,~\cite{BS03, BS04,NO13}):
\begin{align}
\mathcal{U}^d&=\{\pmb{D}=(D_t)_{t\in[T]}\,:\,D_t\leq D_{t+1},
D_t\in [\widehat{D}_t-\Delta_t,\widehat{D}_t+\Delta_t], |\{t\,:\, D_t\neq \widehat{D}_t\}|\leq \Gamma^d
\},\label{dr}\\
\mathcal{U}^c&=\{\pmb{D}=(D_t)_{t\in[T]}\,:\,D_t\leq D_{t+1},
D_t\in [\widehat{D}_t-\Delta_t,\widehat{D}_t+\Delta_t],\parallel \pmb{D}-\widehat{\pmb{D}}\parallel _1 \leq \Gamma^c
\}. \label{cr}
\end{align}
The first representation~(\ref{dr}) is called the
\emph{discrete budgeted uncertainty}, where  $\Gamma^d\in  \{0\} \cup [T]$, and 
the second one~(\ref{cr}) is called the \emph{continuous budgeted uncertainty}, where $\Gamma^c\in \Rset_{+}$. The parameters $\Gamma^d$ and $\Gamma^c$, called \emph{budgets} or \emph{protection levels},
control the amount of uncertainty in $\mathcal{U}^d$ and $\mathcal{U}^c$, respectively. If $\Gamma^d=\Gamma^c=0$, then all cumulative demands take their nominal values (there is only one scenario). On the other hand, for sufficiently large $\Gamma^d$ and $\Gamma^c$ the uncertainty sets $\mathcal{U}^d$ and $\mathcal{U}^c$ are the Cartesian products of $[\widehat{D}_t-\Delta_t, \widehat{D}_t+\Delta_t]$, $t\in [T]$, which yields the interval uncertainty representation discussed in~\cite{KY97}.

In order to compute a \emph{robust production plan}, we adopt the \emph{minmax} 
approach (see, e.g.,~\cite{KY97}), in which we seek a plan minimizing the maximal total cost over all cumulative demand scenarios.
This leads to the 
following minmax  problem:
\begin{equation}
\textsc{MinMax}:\; OPT=\min_{\pmb{x}\in \Xset}\max_{\pmb{D}\in \mathcal{U}} 
\sum_{t\in [T]} \max\{f_I(X_t,D_t), f_B(X_t,D_t) \} ,
     \label{pmm}
\end{equation}
where $\mathcal{U}\in \{\mathcal{U}^d,\mathcal{U}^c\}$, i.e. to the one of computing
 an optimal production plan~$\pmb{x}^*$   of~(\ref{pmm})   along with its worst-case cost~$OPT$.
Indeed, $\pmb{x}^*$ is a robust choice, because we are sure that it optimizes against all scenarios
in $\mathcal{U}$
in which  the amount of uncertainty allocated by  an adversary to  cumulative demands is upper
bounded by a budget  provided. Furthermore, the budget enables to control the level of robustness of a 
production plan computed. More specifically,
an optimal production plan to  the \textsc{MinMax} problem under $\mathcal{U}^d$  optimizes against all scenarios,
in which at most $\Gamma^d$ cumulative demands take values different from their nominal ones at the same time.
Moreover, by changing the value of~$\Gamma^d$, from~$0$ to~$T$,
 one can flexibly control the level of robustness of the  plan computed.
 An optimal production plan 
to the \textsc{MinMax} problem  under $\mathcal{U}^c$  optimizes against all scenarios
in $\mathcal{U}^c$
in which
the total deviation of the
cumulative demands from their nominal values, at the same time, is at most 
a  bound  on the total variability~$\Gamma^c$.
In this case  one can also flexibly control the level of robustness of the  plan computed
 by changing the value of~$\Gamma^c$ from~$0$ to a big number, say $\sum_{t\in [T]} \Delta_t$.

The \textsc{MinMax} problem contains the inner
 \emph{adversarial problem}, i.e.
\begin{equation}
\textsc{Adv}: \; \max_{\pmb{D}\in \mathcal{U}} 
\sum_{t\in [T]}\max\{f_I(X_t,D_t), f_B(X_t,D_t) \} .
\label{padv}
\end{equation}
The \textsc{Adv} problem consists in finding a cumulative demand scenario $\pmb{D}\in \mathcal{U}$ that
maximizes the cost of  a given production plan $\pmb{x}\in \Xset$ over scenario set~$\mathcal{U}$.
In other words, an \emph{adversary}  maliciously wants to increase the cost of~ $\pmb{x}$.

 Throughout this paper, we study the \textsc{Adv} and  \textsc{MinMax} problems under two standing assumptions
 about cumulative demand uncertainty intervals~$\widetilde{D}_t$, $t\in[T]$.
 Under the first one, the intervals 
  are \emph{non-overlapping}, i.e.
$\widehat{D}_{t}+\Delta_t\leq \widehat{D}_{t+1}-\Delta_{t+1}$, $t\in [T-1]$.
This assumption is realistic, in particular at the tactical level of planning, for instance
 in the master production scheduling (MPS)  (see, e.g.~\cite{ACC11}),
where the lengths of periods are big enough (for example, they are equal to one month).
This restriction leads to more efficient methods of solving the problems under consideration. 
Notice that it allows us to drop the constraints $D_t\leq D_{t+1}$, $t\in [T-1]$, in the definition of scenario sets $\mathcal{U}^d$ and $\mathcal{U}^c$.
Under the second assumption,  called the \emph{general case},  we impose no restrictions on the interval bounds, 
i.e. they can now overlap.

\section{Discrete budgeted uncertainty}
\label{sdbu}

In this section we consider the problem with uncertainty set~$\mathcal{U}^d$ defined as~(\ref{dr}). We will discuss the \textsc{Adv} and the \textsc{MinMax} problems, i.e. the problems of evaluating a given production plan and computing the robust production plan, respectively. We provide polynomial algorithms for solving both problems.

\subsection{Non-overlapping case}

In this section we consider the non-overlapping case, i.e we assume that $\widehat{D}_t+\Delta_t\leq \widehat{D}_{t+1}-\Delta_{t+1}$ for each $t\in [T-1]$. We first focus on the  \textsc{Adv} problem, i.e the inner one of  \textsc{MinMax}.
\begin{lem} 
The \textsc{Adv} problem under $\mathcal{U}^d$, for the non-overlapping case, 
boils down to the following problem:
\begin{align}
 \max & \sum_{t\in [T]} \max\{f_I(X_t,\widehat{D}_t-\delta_t\Delta_t), f_B(X_t,\widehat{D}_t+\delta_t\Delta_t) \} \label{aud1}\\
\text{\rm s.t. }& \sum_{t\in [T]} \delta_t\leq \Gamma^d, & \label{aud2}\\
                &0\leq \delta_t\leq 1,\;\;t\in [T].\label{aud3}
\end{align}
Moreover (\ref{aud1})-(\ref{aud3}) has an integral optimal solution $\pmb{\delta}^*$ such
that $\sum_{t\in [T]} \delta^{*}_t= \Gamma^d$.
\label{ladud}
\end{lem}
\begin{proof}
The objective~(\ref{aud1}) is a convex function with respect to~$\pmb{\delta}\in [0,1]^T$. Hence,
it
attains the maximum value at a vertex of the convex  polytope (\ref{aud2})-(\ref{aud3})  (see, e.g.,~\cite{M75}). Moreover, since
 the matrix of constraints~(\ref{aud2}) and
(\ref{aud3}) is totally unimodular and~$\Gamma^d\in\Zset_{+}$,  an optimal solution~$\pmb{\delta}^*$ to (\ref{aud1})-(\ref{aud3}) 
is integral and
has exactly $\Gamma^d$
components equal to~1. 

Let $\pmb{D}^{'}\in \mathcal{U}^d$ be an optimal solution of the \textsc{Adv} problem.
This solution can be expressed by a feasible solution~$\pmb{\delta}^{'}$ to (\ref{aud1})-(\ref{aud3}).
Indeed, $D^{'}_t=\widehat{D}_t\pm \delta^{'}_t\Delta_t$, $t\in[T]$, $\pmb{\delta}^{'}\in [0,1]^T$.
Obviously $\sum_{t\in [T]} \delta^{'}_t\leq \Gamma$, since $|\{t\,:\,D_t\neq \widehat{D}_t\}|\leq \Gamma^d$.
The value of the objective function of~(\ref{padv}) for $\pmb{D}^{'}$ can be bounded from above by
the value of~(\ref{aud1}) for $\pmb{\delta}^{'}$ and, in consequence, by  the value of~(\ref{aud1}) for $\pmb{\delta}^{*}$.
Let us form the cumulative demand scenario~$\pmb{D}^{*}$, that corresponds to~$\pmb{\delta}^{*}$,
 by setting 
 \begin{equation}
D^{*}_t=
\begin{cases}
\widehat{D}_t-\delta^{*}_t\Delta_t  &\text{if
$f_I(X_t,\widehat{D}_t-\delta^{*}_t\Delta_t)>f_B(X_t,\widehat{D}_t+\delta^{*}_t\Delta_t),$}\\
\widehat{D}_t+\delta^{*}_t\Delta_t&\text{otherwise}.
\end{cases}
\label{optadvd}
\end{equation}
Since $\widehat{D}_t+\Delta_t\leq \widehat{D}_{t+1}-\Delta_{t+1}$ for each $t\in [T-1]$ (by the assumption that we consider the non-overlapping case), we get $D^*_t\leq D^*_{t+1}$ and thus $\pmb{D}^{*}\in \mathcal{U}^d$.
We see  at once that
the optimal value of~(\ref{aud1}) for $\pmb{\delta}^{*}$ is equal to 
the value of the objective function of~(\ref{padv}) for~$\pmb{D}^{*}$.
By the optimality of $\pmb{D}^{'}$, this value is bounded from above by 
the value of the objective function of~(\ref{padv}) for $\pmb{D}^{'}$.
Hence  $\pmb{D}^{*}$ is an optimal solution to~(\ref{padv}) as well, which proves the lemma.
\end{proof}
From Lemma~\ref{ladud} and the integrality of~$\pmb{\delta}^*$, it follows that  an optimal solution~$\pmb{D}^{*}$  to~(\ref{padv})
is such that $D_t^{*}\in \{ \widehat{D}_t-\Delta_t, \widehat{D}_t,   \widehat{D}_t+\Delta_t\}$
for every $t\in[T]$. Accordingly, 
we can  provide an algorithm for finding an optimal solution to (\ref{aud1})-(\ref{aud3}).
An easy computation shows that 
the objective function~(\ref{aud1}) can be rewritten as follows:
\begin{equation}
\sum_{t\in [T]}\max\{f_I(X_t,\widehat{D}_t), f_B(X_t,\widehat{D}_t) \} +\sum_{t\in [T]} c_t\delta_t, \label{advfd}
\end{equation}     
where $$c_t=\max\{f_I(X_t,\widehat{D}_t-\Delta_t), f_B(X_t,\widehat{D}_t+\Delta_t) \}-\max\{f_I(X_t,\widehat{D}_t), f_B(X_t,\widehat{D}_t) \},\; t\in [T],$$ 
are fixed coefficients.
The first sum in~(\ref{advfd}) is constant. 
Therefore, in order to solve \textsc{Adv}, we need to solve the following problem:
\begin{align}
 \max &\sum_{t\in [T]} c_t\delta_t \label{caud1}\\ 
 \text{s.t. }& \sum_{t\in [T]} \delta_t\leq \Gamma^d, & \label{caud2}\\
                &0\leq \delta_t\leq 1,\;\;t\in [T].\label{caud3}
\end{align}
Problem (\ref{caud1})-(\ref{caud3}) can be solved in~$O(T)$ time. Indeed,
we first find, in $O(T)$ time (see, e.g.,~\cite{KV12}),
the $\Gamma^d$th largest coefficient, denoted by~$c_{\sigma(\Gamma^d)}$, such that $c_{\sigma(1)}\geq\cdots\geq c_{\sigma(\Gamma^d)}\geq
\cdots\geq c_{\sigma(T)}$, where $\sigma$ is a permutation of~$[T]$.
Then   having $c_{\sigma(\Gamma^d)}$ we can choose  $\Gamma^d$ coefficients~$c_{\sigma(i)}$, $i \in [\Gamma^d]$,
and set $\delta^{*}_{\sigma(i)}=1$. Having the optimal solution $\pmb{\delta}^*$, we can construct the corresponding scenario $\pmb{D}^*$ as in the proof of Lemma~\ref{ladud} (see~(\ref{optadvd})).
Hence, we get the following theorem.
\begin{thm}
The \textsc{Adv} problem  under $\mathcal{U}^d$, for the non-overlapping case,
 can be solved in $O(T)$ time.
\end{thm}

We now show how to solve the \textsc{MinMax} problem for the non-overlapping case in polynomial time. We will first reformulate \textsc{Adv} as a linear programming problem with respect to $X_t$. Then, the linearity will be preserved by adding linear constraints for $X_t$.
Writing the dual to~(\ref{caud1})-(\ref{caud3}), we obtain:
\begin{align}
 \min \;&\Gamma^d \alpha+\sum_{t\in [T]} \gamma_t& \label{dcaud1}\\ 
 \text{s.t. }& \alpha+\gamma_t\geq c_t, & t\in[T], \label{dcaud2}\\
                &\gamma_t\geq 0, &t\in [T], \label{dcaud3}\\
                &\alpha\geq 0,& \label{dcaud4}
\end{align}
where $\alpha$ and $\gamma_t$ are dual variables.
Using (\ref{dcaud1})-(\ref{dcaud4}), equality~(\ref{advfd}), and the definition of $c_t$, we can rewrite (\ref{aud1})-(\ref{aud3}) as:
\begin{align}
 \min \;&\sum_{t\in [T]} \pi_t+\Gamma^d \alpha+\sum_{t\in [T]} \gamma_t& \label{ddcaud1}\\ 
 \text{s.t. }& \pi_t\geq f_I(X_t, \widehat{D}_t), & t\in[T], \label{ddcaud2}\\
                & \pi_t\geq f_B(X_t, \widehat{D}_t), & t\in[T], \label{ddcaud3}\\
                & \alpha+\gamma_t\geq  f_I(X_t,\widehat{D}_t-\Delta_t)-\pi_t, & t\in[T], \label{ddcaud6}\\
                & \alpha+\gamma_t\geq   f_B(X_t, \widehat{D}_t+\Delta_t)-\pi_t, & t\in[T], \label{ddcaud7}\\
                &\alpha,\gamma_t\geq 0,&t\in [T], \label{ddcaud10}\\
                &\pi_t \text{ unrestricted},&t\in [T], \label{ddcaud11}
\end{align}
where constraints~(\ref{ddcaud2})-(\ref{ddcaud3}) specify the maximum operators in the first sum of~(\ref{advfd}) and the
remaining constraints model~(\ref{dcaud2})-(\ref{dcaud4}) and the coefficients~$c_t$, $t\in [T]$. We now prove that (\ref{ddcaud1})-(\ref{ddcaud11}) 
solves the \textsc{Adv} problem.

 \begin{lem}
The program~(\ref{ddcaud1})-(\ref{ddcaud11}) solves the \textsc{Adv} problem.
 \end{lem}
\begin{proof}
It is enough to show that there is an optimal solution $(\pmb{\pi}^*,\pmb{\gamma}^*,\alpha^{*})$ to (\ref{ddcaud1})-(\ref{ddcaud11}) such that 
$\pi^*_t=\max\{ f_I(X_t, \widehat{D}_t), f_B(X_t, \widehat{D}_t)\}$ for each $t\in[T]$.
Let us fix $t\in [T]$ and consider the term $\pi^*_t+\gamma^*_t$ in the objective function. In an optimal solution, we have 
\begin{equation}
\label{eee0}
\gamma^*_t=\max\{[0,f_I(X_t, \widehat{D}_t-\Delta_t)-\pi^*_t-\alpha^*]_{+}, [0, f_B(X_t, \widehat{D}_t+\Delta_t)-\pi^*_t-\alpha^*]_{+}\},
\end{equation}
where  $[y]_{+}=\max\{0,y\}$.
Suppose $\pi^*_t>f_I(X_t, \widehat{D}_t)$ and $\pi_t^*>f_B(X_t, \widehat{D}_t)$. We can then fix $\pi^*_t:=\pi^*_t-\epsilon$, for some $\epsilon>0$, so that $\pi^*_t=\max\{f_I(X_t, \widehat{D}_t), f_B(X_t, \widehat{D}_t)\}$.  Accordingly,  we modify $\gamma^*_t$ using~(\ref{eee0}), preserving the feasibility of (\ref{ddcaud1})-(\ref{ddcaud11}) . The new value of $\gamma_t^*$ is increased by at most $\epsilon$. In consequence, the objective value~(\ref{ddcaud1}) does not increase.
\end{proof}
Adding linear constraints $X_t=\sum_{i\in [t]} x_i$, $t\in [T]$, and  $\pmb{x}\in\Xset$ to  (\ref{ddcaud1})-(\ref{ddcaud11}) 
and using the fact that $f_I$ and $f_B$ are linear with respect to $X_t$
yield  a linear program for 
the \textsc{MinMax} problem. This leads to the following theorem.
\begin{thm}
The \textsc{MinMax} problem   under $\mathcal{U}^d$, for the non-overlapping case,
 can be solved in polynomial time.
\end{thm}

\subsection{General case}

In this section we drop the assumption that the cumulative demand intervals are non-overlapping. 
We first consider the \textsc{Adv} problem.
The following lemma is the key to constructing algorithms in this section, namely it
shows that it is enough to consider only $O(T)$ values of the cumulative demand in each interval.
\begin{lem}
\label{lemvert}
	There is an optimal solution $\pmb{D}^*\in \mathcal{U}^d$ to the \textsc{Adv} problem such that 
	\begin{equation}
D^*_k\in\mathcal{D}_k= [\widehat{D}_k-\Delta_k,\widehat{D}_k+\Delta_k]\cap \bigcup_{t\in[T]} 
\{\widehat{D}_t-\Delta_t,\widehat{D}_t,\widehat{D}_t+\Delta_t\},\; k\in [T].
\label{pves1}
\end{equation}
\end{lem}
\begin{proof}
Scenario set $\mathcal{U}^d$ is not convex. However, it can be decomposed into a union of convex sets in the following way.
Let $\mathcal{T}(\Gamma^d)=\{\mathcal{I}\subseteq [T]:\,|\mathcal{I}|=\Gamma^d\}$ and
define 
\begin{align}
\mathcal{U}^d(\mathcal{I})=\{\pmb{D}=(D_t)_{t\in[T]}\,:\,&
(\forall t\in [T-1])
(D_t\leq D_{t+1}),
(\forall t\in \mathcal{I})
(D_t\in [\widehat{D}_t-\Delta_t,\widehat{D}_t+\Delta_t]), \nonumber\\
&(\forall t\in [T]\setminus \mathcal{I})
(D_t=\widehat{D}_t)  \} \subset \mathcal{U}^d,\label{eee1}
\end{align}
where $\mathcal{I}\in\mathcal{T}(\Gamma^d)$.
Obviously, $\mathcal{U}^d(\mathcal{I})$ is a convex polytope.
Furthermore, it is easy to check that 
$\mathcal{U}^d=\bigcup_{\mathcal{I}\in\mathcal{T}(\Gamma^d)} \mathcal{U}^d(\mathcal{I})$.
The objective function in (\ref{padv})  is convex  with respect to~$\pmb{D}\in \Rset^T_{+}$ and 
 attains the maximum value at a vertex of  a convex  polytope
  (see, e.g.,~\cite{M75}).
  Hence, and from the fact  that $\mathcal{U}^d\not=\emptyset$,
   there exists $\mathcal{I}^*\in  \mathcal{T}(\Gamma^d)$ and, consequently, 
   the convex   polytope $\mathcal{U}^d(\mathcal{I}^*)$,  whose vertex is
an optimal solution to the \textsc{Adv} problem. 

To simplify notation, let us write $\underline{D}_t=\widehat{D}_t-\Delta_t$,
$\overline{D}_t=\widehat{D}_t+\Delta_t$ if  $t\in  \mathcal{I}^*$; and
$\underline{D}_t=\widehat{D}_t$, $\overline{D}_t=\widehat{D}_t$ 
if  $t\in  [T]\setminus\mathcal{I}^*$.
 The constraints $D_t\leq D_{t+1}$, $t\in [T-1]$,
 imply the polytope $\mathcal{U}^d(\mathcal{I}^*)$ does not change if we narrow the intervals so that
 \begin{equation}
 \underline{D}_t\leq \underline{D}_{t+1} \text{ and } \overline{D}_t\leq \overline{D}_{t+1}, \; t\in [T-1].
 \label{cboud}
 \end{equation}
 From now on, we assume that the modified bounds $\underline{D}_t$ and $\overline{D}_t$ fulfill~(\ref{cboud}). Observe that after the narrowing, we get $\underline{D}_k, \overline{D}_k\in \mathcal{D}_k$.

 Assume that $\pmb{D}=(D_t)_{t\in[T]}$
is  a vertex of~$\mathcal{U}^d(\mathcal{I}^*)$ and so $D_t\in [\underline{D}_t,\overline{D}_t]$, $t\in [T]$.
Suppose that $D_k\in (\underline{D}_k, \overline{D}_k)$ for some $k\in \mathcal{I}^*$. Let $(D_q,\dots, D_k,\dots, D_r)$, $q\leq k \leq r$, be subsequence of $(D_1,\dots,D_T)$, such that
$q=\min\{ t\in [T]: D_t=D_k\}$ and $r=\max\{ t\in [T]: D_t=D_k\}$. 
We now claim that $D_q=\overline{D}_{q}$ or $D_r=\underline{D}_{r}$, in consequence
$D_k=\overline{D}_{q}$ or $D_k=\underline{D}_{r}$, which completes the proof.

Suppose, contrary to our claim, that
the first and the last components of the subsequence are  such that
$D_q\in [\underline{D}_{q},\overline{D}_{q})$ and $D_r\in (\underline{D}_{r},\overline{D}_{r}]$.
Observe that $D_q\neq \underline{D}_q$, because otherwise $\underline{D}_q=D_q=D_k>\underline{D}_k$, which contradicts~(\ref{cboud}). Similarly, $D_r\neq \overline{D}_r$.  Therefore,
$D_q\in (\underline{D}_{q},\overline{D}_{q})$ and $D_r\in (\underline{D}_{r},\overline{D}_{r})$.
We thus have
$$D_1\leq D_2\leq D_{q-1}<D_q=D_{q+1}=\dots =D_k=\dots =D_{r-1}=D_r<D_{r+1}\leq\dots\leq D_T.$$
Let $\pmb{\epsilon}\in \Rset^T\setminus\{ \pmb{0}\}$ be such that $\epsilon_t=0$ for each $t\in [T]\setminus\{q,\dots,r\}$ and $\epsilon_t=\sigma>0$ for $t\in\{q,\dots, r\}$. Since $D_i\in (\underline{D}_i, \overline{D}_i)$ for each $i=q,\dots,r$, there is sufficiently small, but positive $\pmb{\epsilon}$, such that 
$\pmb{D}-\pmb{\epsilon} \in \mathcal{U}^d(\mathcal{I}^*)$ and  $\pmb{D}+\pmb{\epsilon}\in \mathcal{U}^d(\mathcal{I}^*)$. Hence $\pmb{D}$ is a strict convex combination of two solutions from $\mathcal{U}^d(\mathcal{I}^*)$ and $\pmb{D}$ is not a vertex solution.
\end{proof}

Equality~(\ref{pves1}) means that, in order to solve \textsc{Adv}, it is enough to examine only $O(T)$ values for 
every cumulative demand $D_k$, $k\in [T]$.
This fact allows us to
transform, in polynomial time,
the \textsc{Adv} problem to a version of the following \emph{restricted longest path problem} (\textsc{RLP} for short). 
We are given 
a layered directed acyclic graph~$G=(V,A)$.
Two nodes $\mathfrak{s}\in V$ and $\mathfrak{t}\in V$ are distinguished as the source node 
(no arc enters to~$\mathfrak{s}$) and the sink node (no arc leaves~$\mathfrak{t}$).
Two attributes $(c_{uw},\delta_{uw})$ are associated with each arc~$(u,w)\in A$, namely
$c_{uw}$ is the length (cost) and $\delta_{uw}$ is the weight of $(u,w)$
and  a bound~$W$  on the total weights of  paths.
In the \textsc{RLP} problem we seek a path~$p$ in~$G$ from  $\mathfrak{s}$ to  $\mathfrak{t}$
whose total weight is at most~$W$ ($\sum_{(u,w)\in p} \delta_{uw}\leq W$)
 and the length (cost) of~$p$  ($\sum_{(u,w)\in p} c_{uw}$) is maximal.

Given an instance of the \textsc{Adv} problem, the corresponding instance of~\textsc{RLP} is constructed as follows.
We first build a layered directed acyclic graph~$G=(V,A)$.
The node set~$V$ is
partitioned into $T+2$ disjoint  layers $V_0,V_1,\ldots,V_T,V_{T+1}$ in which 
$V_0=\{\mathfrak{s}\}$ and $V_{T+1}=\{\mathfrak{t}\}$ contain the source and the sink node, respectively.
Each node $u\in V_k$, $k\in [T]$, corresponds to exactly
one possible value of the $k$th component of a cumulative demand  vertex scenario (see (\ref{pves1})), 
denoted by~$D_{u}$,  
$D_{u}\in \mathcal{D}_k$,
$|V_k|=|\mathcal{D}_k|$.
We also partition~the arc set~$A$ into $T+1$ disjoint subsets, $A= A_1\cup\cdots\cup A_T \cup A_{T+1}$. 
Arc~$(u,w)\in A_1$ if $u\in V_0$ and $w\in V_1$; 
$(u,w)\in A_{T+1}$ if $u\in V_T$ and $w\in V_{T+1}$; and
arc~$(u,w)\in A_{k}$, $k=2,\ldots,T$, if
 $u\in V_{k-1}$, $w\in V_{k}$ and $D_{u}\leq D_{w}$.
 Set $W=\Gamma^d$.
Finally two attributes are associated with 
each arc $(u,w)\in A$: $c_{uw}$ and $\delta_{uw}\in \{0,1\}$, whose values are determined as follows:
\begin{equation}
(c_{uw},\delta_{uw})=
\begin{cases}
(\max\{f_I(X_k,D_w), f_B(X_k,D_w)\},1)&\text{if $(u,w)\in A_k$, $D_w\not=\widehat{D}_k$, $k\in [T]$,}\\
(\max\{f_I(X_k,D_w), f_B(X_k,D_w)\},0)&\text{if $(u,w)\in A_k$, $D_w=\widehat{D}_k$, $k\in [T]$,}\\
        (0,0)   &\text{if  $(u,w)\in A_{T+1}$}.
\end{cases}
\label{cuvduv}
\end{equation}
The second atribute $\delta_{uw}$ is equal to 1 if the value of~$D_w$
 differs from the nominal value $\widehat{D}_k$, and
 0 otherwise, $k\in [T]$. The transformation can be done in $O(T^3)$ time, since
 $A$ has $O(T^3)$ arcs. An example is shown in Figure~\ref{fig1}. At the nodes, other than $\mathfrak{s}$ and $\mathfrak{t}$, the possible values of the cumulative demands $D_u$ are shown. Observe that $\bigcup_{t\in[T]} 
\{\widehat{D}_t-\Delta_t,\widehat{D}_t,\widehat{D}_t+\Delta_t\}=\{1,2,3,5,6,7,9\}$. Hence $D_1\in
\mathcal{D}_1=\{1,3,5\}$, $D_2\in \mathcal{D}_2=\{3,5,6,7,9\}$ and $D_3\in
\mathcal{D}_3=\{5,6,7\}$. We can further assume that $D_2\neq 9$, due to the constraint $D_2\leq D_3$.
 \begin{figure}[ht]
 \centering
 \includegraphics[scale=0.35]{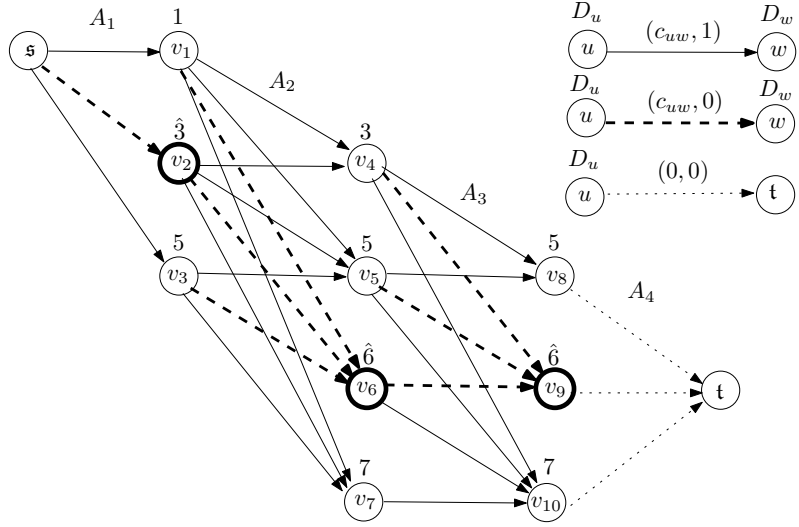}
 \caption{Graph for $T=3$ and $D_1\in [\hat{3}-2, \hat{3}+2]$, $D_2\in [\hat{6}-3,\hat{6}+3]$, $D_3\in [\hat{6}-1, \hat{6}+1]$. Notice that $D_2\leq 7$ for any feasible cumulative demand scenario. If $\Gamma^d=2$, then we seek a longest 
 $\mathfrak{s}$-$\mathfrak{t}$ path using at most~2 solid arcs.} \label{fig1}
 \end{figure}

 \begin{prop}
 A cumulative demand scenario with the cost of~$C^*$ is optimal for an instance of the \textsc{Adv} problem
 if and only if there is an optimal $\mathfrak{s}$-$\mathfrak{t}$ path with the length (cost) of~$C^*$
 for the constructed instance of
  the \textsc{RLP} problem.
  \label{padvrlp}
 \end{prop}
 \begin{proof}
 Suppose that $\pmb{D}^*=(D^*_k)_{k\in[T]}\in \mathcal{U}^d$ with the cost of~$C^*$
 is an  optimal cumulative demand scenario  for an instance of~\textsc{Adv}
  for $\mathcal{U}^d$. By Lemma~\ref{lemvert},
 without loss of generality, we can assume that  $\pmb{D}^*$ is a  vertex of $\mathcal{U}^d(\mathcal{I})$ for
 some  $\mathcal{I}\in  \mathcal{T}(\Gamma^d)$.
 Thus $D^*_k\in \mathcal{D}_k$, $k\in [T]$, (see~(\ref{pves1})).
 From the construction of~$G$ and the definition of~$\mathcal{T}(\Gamma^d)$
 it follows that $\pmb{D}^*$ corresponds to
 an $\mathfrak{s}$-$\mathfrak{t}$ path in~$G$, say~$p^*$, which satisfies the budget constraint 
 $\sum_{(u,w)\in p^*} \delta_{uw}\leq \Gamma^d$ and with the length of~$C^*$.
 We claim that $p^*$ is an optimal path  for~\textsc{RLP} in~$G$.
 Suppose, contrary to our claim, that there exists a feasible $\mathfrak{s}$-$\mathfrak{t}$ path~$p'$ in~$G$
 of length (cost) greater than~$C^*$. By the construction of~$G$, the first $T$ arcs of~$p'$
 correspond to scenario~$\pmb{D}'=(D'_k)_{k\in[T]}$ 
 such that  $D'_k\in \mathcal{D}_k$, $k\in [T]$, and 
  $D'_t\leq D'_{t+1}$. Obviously, $\pmb{D}'$ has the same cost as~$p'$.
 Since $\sum_{(u,w)\in p'} \delta_{uw}\leq \Gamma^d$, $\pmb{D}'\in \mathcal{U}^d(\mathcal{I})$ for
 some  $\mathcal{I}\in  \mathcal{T}(\Gamma^d)$ and in consequence $\pmb{D}'\in\mathcal{U}^d$.
 This contradicts the optimality of $\pmb{D}^*$ over set $\mathcal{U}^d$.
 
 Assume that $p^*$ is an optimal  $\mathfrak{s}$-$\mathfrak{t}$ path with the length (cost) of~$C^*$
 in~$G$. Similarly,
 from the construction of~$G$ and the feasibility of~$p^*$ it follows that its first $T$~arcs
 correspond to scenario~$\pmb{D}^*$ from $\mathcal{U}^d(\mathcal{I})$ for
 some  $\mathcal{I}\in  \mathcal{T}(\Gamma^d)$ and  from $\mathcal{U}^d$. 
 The cost of  $\pmb{D}^*$ is equal to the cost of~$p^*$.
Furthermore, (\ref{pves1}) and again the construction of~$G$ show that
each vertex scenario of~$ \mathcal{U}^d(\mathcal{I})$ for every
$\mathcal{I}\in \mathcal{T}(\Gamma)$ corresponds to  a feasible $\mathfrak{s}$-$\mathfrak{t}$ path in~$G$
with the same cost, that is not better than~$C^*$ due to  the   optimality of~$p^*$.
Lemma~\ref{lemvert} now leads to  the optimality of~$\pmb{D}^*$
for an instance of~\textsc{Adv}.
  \end{proof}

 In general, the restricted longest (shortest) path problem is weakly NP-hard, even for series-parallel graphs
(see, e.g.,~\cite{GJ79}). However, it can be solved in pseudopolynomial time
$O(|A|W)$ in directed acyclic graphs, if 
the bound~$W\in \Zset_{+}$ and the weights~$\delta_{uw}\in  \Zset_{+}$, $(u,v)\in A$,
(see, e.g.,~\cite{H92}). Fortunately,  in our case 
$W=\Gamma^d\leq T$ and~$A$ has
 $O(T^3)$ arcs. We are thus led to the following theorem.
\begin{thm}
The \textsc{Adv} problem  for $\mathcal{U}^d$  
 can be solved in $O(T^4)$ time.
\end{thm}

We now deal with the problem of computing an optimal production plan. Our goal is to construct a linear programming formulation for the \textsc{MinMax} problem. 
Unfortunately,
a direct approach based on the network flow theory leads to
 a linear program with the associated  polytope that has not the integrality property (see~\cite{DR00}), i.e the one being the intersection of 
 the path polytope defined as
the convex hull of the characteristic vectors of 
$\mathfrak{s}$-$\mathfrak{t}$ paths in~$G$ and the half-space defined by
the budged constraint $\sum_{(u,w)\in p} \delta_{uw}\leq W$,  even if  $W=\Gamma^d$  is  bounded by an integer~$T$
and  the weights~$\delta_{uw}\in  \{0,1\}$, $(u,v)\in A$.
 In consequence, such a restricted longest path problem in~$G$ cannot be solved as 
the flow based linear program.
However,  using the fact that $W=\Gamma^d\leq T$, we will  transform in polynomial time our RLP problem to the longest path problem in acyclic graphs, for which a compact linear program can be built. The idea consists in transforming $G=(A,V)$ into $G'=(A',V')$ by splitting each node~$u$ of $G$, different from~$\mathfrak{t}$, into at most $\Gamma^d+1$ nodes labeled as $u^0,\dots,u^{\Gamma^d}$. Each arc $(u,w)\in A$ with
 the attributes $(c_{uw}, \delta_{uw})$ (see~(\ref{cuvduv})) induces the set of arcs $(u^j,w^{j+\delta_{uv}})$, $j=0,\dots,\Gamma^d$, $j+\delta_{uw}\leq \Gamma^d$ in $A'$ with the  same costs $c_{uw}$. The nodes from the $T$th layer are connected to $\mathfrak{t}$ by arcs with~0 cost.
 We remove from $G'$ all nodes which are not connected to $\mathfrak{s}^0$ or $\mathfrak{t}$, obtaining a reduced graph~$G$'. The resulting graph~$G'$, for the graph~$G$ presented in Figure~\ref{fig1}, is shown in Figure~\ref{fig2}.
 \begin{figure}[ht]
 \centering
 \includegraphics[scale=0.27]{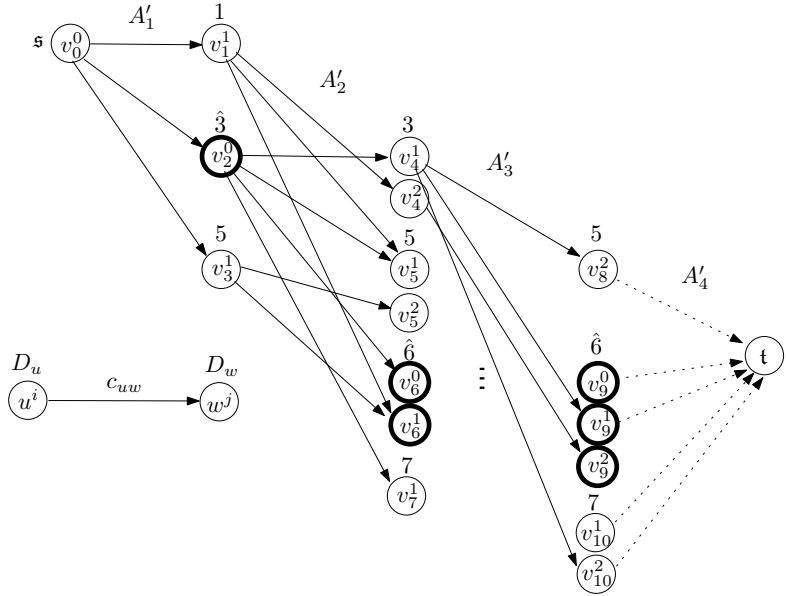}
 \caption{Graph $G'$ corresponding to $G$ presented in Figure~\ref{fig1}. 
 For readability
 not all arcs from the $A'_3$ are shown.} \label{fig2}
 \end{figure}
 
Graph~$G'$ has  $O(T^2\Gamma^d)$ nodes and  $O(T^3 \Gamma^{d})$ arcs. Since $\Gamma^d\leq T$, its size is polynomial in the input size of the \textsc{MinMax} problem. We now use the dual linear programming formulation of the longest path problem in $G'$. Let us associate unrestricted variable $\pi_u^j$ with each node $u^j$ of $G'$ other than $\mathfrak{t}$ and unrestricted variable $\pi_{\mathfrak{t}}$ with node $t$. The corresponding linear programming formulation for the longest path problem in $G'$ (and for $\textsc{Adv}$) is a follows:
$$
\begin{array}{lll}
	\min &\pi_{\mathfrak{t}} \\
	\text{s.t.} & \pi_{\mathfrak{s}}^0=0,\\
		 & \pi_{w}^j-\pi_{u}^i\geq c_{uw}, & (u^i, w^j)\in A', \\
	 & \pi_{\mathfrak{t}}-\pi_w^j \geq 0,& (w^j, \mathfrak{t})\in A'.
\end{array}
$$
Using the definition of $c_{uw}$ (see~(\ref{cuvduv})), and 
adding linear constraints $X_t=\sum_{i\in [t]} x_i$, $t\in [T]$, and  $\pmb{x}\in\Xset$,
we can rewrite this program as follows:
\begin{equation}
\begin{array}{lll}
	\min &\pi_{\mathfrak{t}} \\
	\text{s.t.} & \pi_{\mathfrak{s}}^0=0,\\
		 & \pi_{w}^j-\pi_{u}^i\geq f_I(X_k,D_w), & (u^i, w^j)\in A',\\
		 & \pi_{w}^j-\pi_{u}^i\geq f_B(X_k,D_w), & (u^i, w^j)\in A',\\
		  & \pi_{\mathfrak{t}}- \pi_w^j\geq 0, & (w^j, \mathfrak{t})\in A',\\
		  & X_t=\sum_{i\in [t]} x_i, & t\in [T],\\
		  & \pmb{x}\in \Xset.
\end{array}
\label{lpminmax}
\end{equation}
Formulation~(\ref{lpminmax}) is a linear programming formulation for the \textsc{MinMax} problem, with $O(T^3)$ variables and  $O(T^4)$ constraints.  We thus get the following result:
\begin{thm}
The \textsc{MinMax} problem   under $\mathcal{U}^d$  in the general case
 is polynomially solvable.
 \end{thm}

\section{Continuous budgeted uncertainty}
\label{scbu}

In this section we discuss the \textsc{MinMax} and \textsc{Adv} problem under the 
continuous  budgeted cumulative demand  uncertainty~$\mathcal{U}^c$ defined in~(\ref{cr}).
Similarly as in Section~\ref{sdbu},
we first study the non-overlapping case and 
 we consider then  the general one.
We provide negative and positive complexity results for the  \textsc{Adv}  and \textsc{MinMax} problems.

\subsection{Non-overlapping case}
We start by analyzing the properties of the \textsc{Adv} problem.
\begin{lem} 
The \textsc{Adv} problem under $\mathcal{U}^c$, for the non-overlapping case,
boils down to the following problem:
\begin{align}
 \max &\sum_{t\in [T]}\max\{ f_I(X_t, \widehat{D}_t-\delta_t), f_B(X_t, \widehat{D}_t+\delta_t)\}
 \label{auc1}\\ 
\text{\rm s.t. }& \sum_{t\in [T]} \delta_t\leq \Gamma^c, & \label{auc2}\\
                &0\leq \delta_t\leq \Delta_t,\;\;t\in [T].\label{auc3}
\end{align}
\label{laduc}
\end{lem}
\begin{proof}
The existence of an optimal vertex solution~$\pmb{\delta}^*$ to (\ref{auc1})-(\ref{auc3}) follows
from  the convexity of~the objective function~(\ref{auc1})  (see, e.g.,~\cite{M75}).
Let $\pmb{D}^{'}\in \mathcal{U}^c$ be an optimal solution of the \textsc{Adv} problem.
Since $\parallel \pmb{D}^{'}-\widehat{\pmb{D}}\parallel _1 \leq \Gamma^c$,
$D^{'}_t=\widehat{D}_t\pm \delta^{'}_t$, where $\delta_t^{'}\in [0,\Delta_t]$, $t\in[T]$, such that $\sum_{t\in [T]} \delta^{'}_t\leq \Gamma^c$.
Thus $\pmb{\delta}^{'}$ is a feasible solution to (\ref{auc2})-(\ref{auc3}).
The value of the objective function of~(\ref{padv}) for $\pmb{D}^{'}$ can be bounded from above by
the value of~(\ref{auc1}) for $\pmb{\delta}^{'}$, and  so by  the value of~(\ref{auc1}) for $\pmb{\delta}^{*}$.
On the other hand,
the cumulative demand scenario~$\pmb{D}^{*}\in \mathcal{U}^c$ corresponding to~$\pmb{\delta}^{*}$ can be built as follows:
\begin{equation}
D^{*}_t=
\begin{cases}
\widehat{D}_t-\delta^{*}_t &\text{if
$f_I(X_t,\widehat{D}_t-\delta^{*}_t)>f_B(X_t,\widehat{D}_t+\delta^{*}_t),$}\\
\widehat{D}_t+\delta^{*}_t&\text{otherwise},
\end{cases}
\label{optapvc}
\end{equation}
and by the optimality of~$\pmb{D}^{'}$,
the value of the objective function of~(\ref{padv}) for $\pmb{D}^{*}$, which equal to 
the value of~(\ref{auc1}) for $\pmb{\delta}^{*}$, is  bounded from above by 
the value  of~(\ref{padv}) for $\pmb{D}^{'}$. Therefore
$\pmb{D}^{*}$ is also  an optimal solution to~(\ref{padv}) and the lemma follows.
\end{proof}
Lemma~\ref{laduc} now shows that solving the \textsc{Adv} problem is equivalent to
solving (\ref{auc1})-(\ref{auc3}). An optimal solution to~\textsc{Adv} can be formed according to~(\ref{optapvc}).
The following problem is an equivalent reformulation of (\ref{auc1})-(\ref{auc3}):
\begin{align}
 \max \;&A+\sum_{t\in [T]} c_t(\delta_t) \label{aucr1}\\
\text{\rm s.t. }& \sum_{t\in [T]} \delta_t\leq \Gamma^c, & \label{aucr2}\\
                &0\leq \delta_t\leq \Delta_t,\;\;t\in [T],\label{aucr3}
\end{align}
where
\[
A=\sum_{t\in[T]} \max\{f_I(X_t, \widehat{D}_t), f_B(X_t,\widehat{D}_t)\}
\]
and
$$
	c_t(\delta)=\max\{f_I(X_t, \widehat{D}_t-\delta), f_B(X_t,\widehat{D}_t+\delta)\}-\max\{f_I(X_t, \widehat{D}_t), f_B(X_t,\widehat{D}_t)\}, t\in [T],
$$
is a linear or piecewise linear nonnegative convex function in $[0,\Delta_t]$, $t\in[T]$,
(see Figure~\ref{fcosts}). Observe that $A$ is constant.
Thus in order to solve the \textsc{Adv} problem we need to solve
 the inner optimization problem of (\ref{aucr1})-(\ref{aucr3}), i.e.
\begin{align}
 \max &\sum_{t\in [T]} c_t(\delta_t) \label{cauc1}\\ 
 \text{s.t. }& \sum_{t\in [T]} \delta_t\leq \Gamma^c, & \label{cauc2}\\
                &0\leq \delta_t\leq \Delta_t,\;\;t\in [T].\label{cauc3}
\end{align}
\begin{figure}
  \centering
  \includegraphics[scale=0.43]{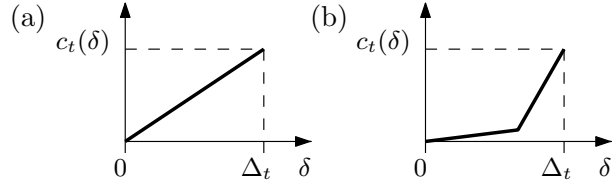}
  \caption{Functions~$c_t(\delta)$, $t\in [T]$, in the  problem (\ref{aucr1})-(\ref{aucr3}).}\label{fcosts}
\end{figure}
The above problem
is a special case of  a continuous knapsack problem with separable convex utilities that
 is  weakly NP-hard in general (see~\cite{LPR14}).
In (\ref{cauc1})-(\ref{cauc3})
the separable convex utilities $c_t(\delta)$ have the simpler forms,  piecewise linear ones, depicted in Figure~\ref{fcosts}.
The next lemma shows that  (\ref{cauc1})-(\ref{cauc3}) is also weakly NP-hard 
even for the restricted separable convex utilities $c_t(\delta)$ shown in Figure~\ref{fcosts}b
and so \textsc{Adv} is weakly NP-hard (the proof is similar in spirit to that in~\cite{LPR14}).
\begin{thm}
	The \textsc{Adv} problem under $\mathcal{U}^c$, for the non-overlapping case, is weakly NP-hard.
	\label{thmcompladv}
\end{thm}
\begin{proof}
Consider the following  weakly NP-hard \textsc{Subset Sum} problem (see, e.g.,~\cite{GJ79}),
in which we are given a collection $\{a_1,\ldots,a_n\}$ of $n$~positive integers and an integer $b$. 
We ask if there is a subset $S\subseteq [n]$ such that $\sum_{i\in S} a_i=b$. 

We first show a polynomial time reduction from \textsc{Subset Sum} to  (\ref{cauc1})-(\ref{cauc3}).
We are given an instance of the \textsc{Subset-Sum} and a  corresponding instance of is built by setting the following 
parameters:
the number of periods $T=n$;  the costs $c^P=0$, $c^I=0$, $c^B=2$; the selling price $b^P=0$;
the nominal  value of cumulative demand $\widehat{D}_t=tA$ for every $t\in[T]$, where $A=\sum_{i\in [n]} a_i$;
the upper bound 
 $\Delta_t=a_t$ for every $t\in [T]$;
 the cumulative production  $X_t=tA+\frac{1}{2}a_t$ for every $t\in[T]$;
 the budget $\Gamma^c=b$. Therefore
 $c_t(\delta)=\max\{0, 2\delta-a_t\}$, $\delta\in [0,a_t]$,
is  a piecewise linear convex utility function, $t\in[T]$. Now the problem (\ref{cauc1})-(\ref{cauc3}) has the following form:
\begin{align}
 \max z=&\sum_{t\in [T]} \max\{0, 2\delta_t-a_t\}  \label{caucnp1}\\ 
 \text{s.t. }& \sum_{t\in [T]} \delta_t\leq b, & \label{caucnp2}\\
                &0\leq \delta_t\leq a_t,\;\;t\in [T].\label{caucnp3}
\end{align}
Note that  $c_t(\delta)=\max\{0, 2\delta-a_t\}=a_t$ if and only if  $\delta=a_t$ and
$c_t(\delta)=\max\{0, 2\delta-a_t\}<\delta$ for $\delta\in (0,a_t)$. Thus, by the constraint~(\ref{caucnp2}), the optimal value of (\ref{caucnp1})
is bounded from above by~$b$.  
Accordingly,  it is easily seen that $z=b$ if and only if  the instance of \textsc{Subset Sum} is positive.
Indeed,
$z=b$ if and only if for each $\delta_t>0$, the equality $\delta_t=a_t$ holds, which means that  the answer to \textsc{Subset Sum} is yes.  This
completes the proof is  (\ref{cauc1})-(\ref{cauc3}) is weakly NP-hard.
Furthermore, a trivial verification shows that (\ref{caucnp1})-(\ref{caucnp3}) is an instance of (\ref{auc1})-(\ref{auc3}.
Hence, the  \textsc{Adv} problem under $\Gamma^c$ is weakly NP-hard as well.
\end{proof}

Before proposing  a pseudopolynomial algorithm for the \textsc{Adv} problem, we show 
 the integrality property of an optimal solution  to  (\ref{cauc1})-(\ref{cauc3})  or equivalently  to~(\ref{auc1})-(\ref{auc3}).
 The following lemma is a key one.
\begin{lem}
Suppose that $\Gamma^c,\Delta_t\in \Zset_{+}$, $t\in [T]$. Then
there exists an optimal solution~$\pmb{\delta}^*$ to~(\ref{cauc1})-(\ref{cauc3})  such that
$\delta^{*}_t\in \{0,1,\ldots,\Delta_t\}$, $t\in [T]$.
\label{ipcb}
\end{lem}
\begin{proof}
Substituting $\sum_{i\in [\Delta_t]} \delta^{i}_t$ into $\delta_t$, where $\delta^{i}_t\in[0,1]$, $t\in [T]$,
we can rewrite  (\ref{cauc1})-(\ref{cauc3})  as:
\begin{align*}
 \max &\sum_{t\in [T]} c_t\left(\sum_{i\in [\Delta_t]} \delta^{i}_t \right) \\ 
 \text{s.t. }& \sum_{t\in [T]} \sum_{i\in [\Delta_t]} \delta^{i}_t\leq \Gamma^c, & \\\
                &0\leq \delta^{i}_t\leq 1,\;\;t\in [T],i\in [\Delta_t].
\end{align*}
The constraint matrix of the resulting problem is totally unimodular.
Hence, be the assumption $\Gamma^c,\Delta_t\in \Zset_{+}$, $t\in [T]$, each vertex solution is integral, i.e. $\delta^{i}_t\in \{0,1\}$, $t\in [T],i\in [\Delta_t]$ (see, e.g.,~\cite{KV12}).
Moreover, since the objective function is convex, 
it
attains its maximum value at a vertex solution (see, e.g.,~\cite{M75}) that is integral.
Thus an original optimal solution~$\pmb{\delta}^*$ is integral as well and the lemma follows.
\end{proof}
\begin{figure}
 \centering
 \includegraphics[scale=0.196]{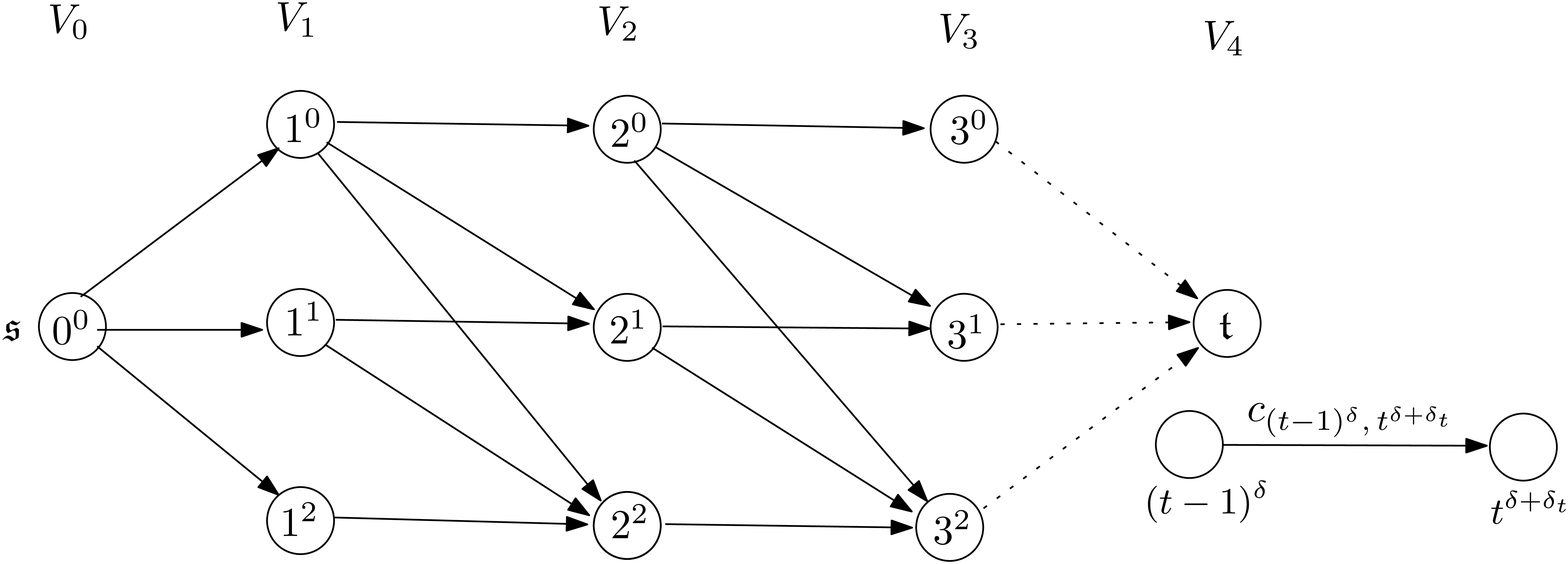}
 \caption{Graph $G$ for $T=3$, $\Gamma^c=2$ and $\Delta_t=2$, $t\in [T]$.} \label{fig3}
 \end{figure}

Now we ready to give a pseudopolynomial transformation of  problem~(\ref{auc1})-(\ref{auc3}) 
to a longest path problem  in a layered directed acyclic graph~$G=(V,A)$. 
We are given an instance 
of  (\ref{auc1})-(\ref{auc3}), where $\Gamma^c,\Delta_t\in \Zset_{+}$, $t\in [T]$.
 Graph~$G=(V,A)$ is build as follows:
the set~$V$ is
partitioned into $T+2$ disjoint  layers $V_0,V_1,\ldots,V_T,V_{T+1}$, in which 
each layer~$V_t$ corresponding to period~$t$, $t\in[T]$, has $\Gamma^c+1$ nodes denoted by
$t^0,\ldots,t^{\Gamma^c}$; sets
$V_0=\{\mathfrak{s}=0^0\}$ and $V_{T+1}=\{\mathfrak{t}\}$ contain two distinguished nodes, $\mathfrak{s}$ and 
$\mathfrak{t}$. 
The notation $t^{\delta}$, $\delta=0,\ldots, \Gamma^c$, means that~$\delta$ units of the
available uncertainty~$\Gamma^c$ have been allocated by an adversary to the cumulative demands
in periods from~$1$ to~$t$.
Each node $(t-1)^{\delta}\in V_{t-1}$, 
$t\in [T]$ (including the source node~$\mathfrak{s}=0^0$ in $V_0$) has at most~$\Delta_{t}+1$ arcs
that go to nodes in layer~$V_{t}$, namely arc $((t-1)^{\delta},t^{\delta+\delta_{t}})$ exists
and it is included to the set of arcs~$A$ if 
$\delta+\delta_{t}\leq \Gamma^c$,
where $\delta_{t}=0,\ldots,\Delta_{t}$. Moreover,
we associate with such arc  $((t-1)^{\delta},t^{\delta+\delta_{t}})\in A$  
the cost $c_{(t-1)^{\delta},\,t^{\delta+\delta_{t}}}$ in the following way (see also~(\ref{auc1})):
\begin{equation}
c_{(t-1)^{\delta},\,t^{\delta+\delta_{t}}}=\max\{ f_I(X_t, \widehat{D}_t-\delta_t), f_B(X_t, \widehat{D}_t+\delta_t)\}.
\label{cijdag}
\end{equation}
Notice that the costs are constant, because $\pmb{x}\in \Xset$ is fixed.
We finish with connecting  each node from~$V_T$ with
 the sink node~$\mathfrak{t}$ by the arc of zero cost.
 The transformation can be done in $O(T\Gamma^c\Delta_{\max})$ time, where $\Delta_{\max}=\max_{t\in [T]}\Delta_t$. An example for $T=3$, $\Gamma^c=2$ and $\Delta_t=2$, $t\in [T]$, is shown in Figure~\ref{fig3}.
 \begin{prop}
 A solution
 with the cost of~$C^*$ is optimal for an instance of problem~(\ref{auc1})-(\ref{auc3}) 
 if and only if there is a longest $\mathfrak{s}$-$\mathfrak{t}$ path with the length  of~$C^*$ in~$G$ constructed.
  \label{padvlp}
 \end{prop}
 \begin{proof}
  A trivial verification shows that each path from $\mathfrak{s}$ to
 $\mathfrak{t}$  in~$G$ (its first~$T$ arcs) 
 models an integral feasible solution~$\pmb{\delta}=(\delta_t)_{t\in[T]}$ to~(\ref{auc1})-(\ref{auc3}) 
 and vise a versa  if $\Gamma^c,\Delta_t\in \Zset_{+}$, $t\in [T]$
 (see Lemma~\ref{ipcb}).
 Indeed, consider any $\mathfrak{s}$-$\mathfrak{t}$ path in~$G$, by the construction of~$G$, its form is as follows:
 $\mathfrak{s}=0^0\leadsto1^{0+\delta_1} \leadsto 2^{(0+\delta_1)+\delta_2} \leadsto
 \cdots \leadsto (t-1)^{(0+\sum_{k\in[t-2]}\delta_k)+\delta_{t-1}}  \leadsto t^{0+(\sum_{k\in[t-1]}\delta_k)+\delta_t} 
 \leadsto  \cdots \leadsto
 (T-1)^{(0+\sum_{k\in[T-2]}\delta_k)+\delta_{T-1}} \leadsto T^{(0+\sum_{k\in[T-1]}\delta_k)+\delta_{T}} \leadsto \mathfrak{t}$, where
 $\delta_1\leq\Delta_1,\delta_2\leq\Delta_2,\ldots,\delta_t\leq\Delta_t,
 \ldots, \delta_T\leq\Delta_T$, $\delta_{t}\in \Zset_{+}$, and 
 its arcs  $((t-1)^{(0+\sum_{k\in[t-2]}\delta_k)+\delta_{t-1}} ,t^{0+(\sum_{k\in[t-1]}\delta_k)+\delta_t})\in A$, since
 $0+(\sum_{k\in[t-1]}\delta_k)+\delta_t\leq \Gamma^c$, $t\in [T]$. Thus
 the total amount of uncertainty  to
 cumulative demands along this path (along its first~$T$ arcs) is at most~$\Gamma^c$, i.e. $\sum_{t\in[T]}\delta_t\leq \Gamma^c$.
 Furthermore, it follows from~(\ref{cijdag}) that the cost of this path is equal to
 the value of
 the objective function~(\ref{auc1}) for $\pmb{\delta}$.
 Accordingly, the costs of an optimal solution to~(\ref{auc1})-(\ref{auc3}) and the length of a longest path in~$G$ are
 the same, equal to~$C^*$.
 \end{proof}
From Lemma~\ref{laduc} and Proposition~\ref{padvlp} it follows that
solving the \textsc{Adv} problem boils down to finding a longest path
from $\mathfrak{s}$ to $\mathfrak{t}$ in~$G$ built, which can be done in $O(|A|+|V|)$ (see, e.g.,~\cite{AMO93}).
Taking into account the running time required to construct~$G$ and finding a longest path in~$G$ we obtain
the following theorem:
\begin{thm}
Suppose that $\Gamma^c,\Delta_t\in \Zset_{+}$, $t\in [T]$. Then
the \textsc{Adv} problem  under $\mathcal{U}^c$, for the non-overlapping case,
 can be solved in $O(T\Gamma^c\Delta_{\max})$ time.
 \label{tcadv}
\end{thm} 
There are some polynomially solvable cases of the \textsc{Adv} problem
(problem~(\ref{auc1})-(\ref{auc3})).
The first one is obvious, namely, when $\Gamma^c$ is bounded by a polynomial of
the problem size (notice that $\Delta_{\max}\leq\Gamma^c$).
The second case is the uniform one, i.e. bounds $\Delta_t=\Delta$ for
every $t\in [T]$, and then
 the inner problem (\ref{cauc1})-(\ref{cauc3}) can be solved  in $O(T^2)$ 
time~\cite{LPR14}. The last case,
when $c_t(\delta)$, in the objective function~(\ref{cauc1}), is linear 
(see Figure~\ref{fcosts}a) for every $t\in [T]$. Then   
the problem (\ref{cauc1})-(\ref{cauc3}) becomes a 
continuous knapsack problem with separable linear utilities that can be solved in $O(T)$ time
(see, e.g.,~\cite{KV12}). 

It turns out that
if $\Gamma^c,\Delta_t\in \Zset_{+}$, $t\in [T]$, then
 the \textsc{Adv} problem has a fully polynomial
approximation scheme 
(FPTAS)\footnote{A maximization (resp. minimization) problem has an FPTAS if for each~$\epsilon>0$ and every its instance~$I$
the inequality $OPT(I)\leq (1+\epsilon) c(I)$
(resp. $c(I)\leq (1+\epsilon) OPT$)
 holds,
where $OPT(I)$ is the optimal cost of~$I$ and $c(I)$ is the cost returned by an approximation algorithm
whose running time is polynomial in both  $1/\epsilon$ and the size of~$I$.
It is assumed that  the cost of each possible solution of the problem is nonnegative.}.
Indeed,
the existence of the FPTAS  follows
from the fact that  problem  (\ref{cauc1})-(\ref{cauc3}) with the integral property is
a special case of the nonlinear knapsack
problem with a separable nondecreasing objective function, a separable nondecreasing
packing (budget) constraint and integer variables that admits an FPTAS (see~\cite{HKLOL14}).
\begin{cor}
Suppose that $\Gamma^c,\Delta_t\in \Zset_{+}$, $t\in [T]$. Then
the \textsc{Adv} problem  under $\mathcal{U}^c$, for the non-overlapping case, admits 
 an FPTAS.
 \label{cfadv}
\end{cor}
\begin{proof}
Let $\pmb{\delta}^*$ be an optimal solution to (\ref{aucr1})-(\ref{aucr3}) (equivalently to  the \textsc{Adv} problem) and
$\pmb{\delta}^{'}$  be a solution to (\ref{cauc1})-(\ref{cauc3}) returned by an FPTAS proposed in~\cite{HKLOL14}.
Obviously, the running time  of the FPTAS for  (\ref{aucr1})-(\ref{aucr3} is the same as  the one for  (\ref{cauc1})-(\ref{cauc3}).
We  only need to show that the inequality
$A+\sum_{t\in [T]} c_t(\delta^{*}_t)\leq (1+\epsilon)(A+\sum_{t\in [T]} c_t(\delta^{'}_t))$ holds for every $\epsilon>0$.
There is no loss of generality in assuming $A\geq 0$. Hence and from the fact that $\pmb{\delta}^{'}$ is an 
approximate solution to (\ref{cauc1})-(\ref{cauc3}), we get
$A+\sum_{t\in [T]} c_t(\delta^{*}_t)\leq A+(1+\epsilon)\sum_{t\in [T]} c_t(\delta^{'}_t)
\leq (1+\epsilon)(A+\sum_{t\in [T]} c_t(\delta^{'}_t))$.
\end{proof}

We now deal with the \textsc{MinMax} problem.  
Theorem~\ref{thmcompladv} immediately 
yields the following corollary. 
\begin{cor}
The \textsc{MinMax} problem under $\mathcal{U}^c$, for the non-overlapping case, is weakly NP-hard.
\label{cmmor}
\end{cor}

We now provide some positive results for the \textsc{MinMax} problem.
We propose an ellipsoid algorithm based approach, adapted from~\cite{ASNP16},
where a similar class of robust problems (with a different scenario set)  
has been studied. The \textsc{MinMax} problem (see (\ref{pmm})) can be formulated as the following 
convex programing model:
\begin{align}
 \min\; & \alpha &   \label{gpmmc1}\\
   \text{s.t. }   &F(\pmb{x},\pmb{D})\leq \alpha, &\pmb{D} \in \mathcal{U}^c, \label{gpmmc2}\\
    &\pmb{x}\in \Xset, & \label{gpmmc3}
\end{align} 
where $F(\pmb{x},\pmb{D})=\sum_{t\in [T]}\max\{f_I(X_t,D_t), f_B(X_t,D_t) \}$ is a convex function (recall that $X_t=\sum_{i\in [t]} x_i$).
Thus the above  program has infinitely many convex constraints of the form~(\ref{gpmmc2}) and
together with the linear constraints~(\ref{gpmmc3}) that describe a convex set. One can solve (\ref{gpmmc1})-(\ref{gpmmc3}) by the ellipsoid algorithm (see,~e.g.,~\cite{GLS93}).
By the equivalence of optimization and separation (see,~e.g.,~\cite{GLS93}), 
we need only a \emph{separation oracle} for the convex set~$\Pset$ determined by the constraints
 (\ref{gpmmc2}) and (\ref{gpmmc3}), i.e. a procedure, which for given 
 $(\pmb{x}^{*}, \alpha^{*})\in \Rset^{T+1}$,
 either decide that  $(\pmb{x}^{*}, \alpha^{*})\in \Pset$ or
 return a \emph{separating hyperplane} between $\Pset$ and $(\pmb{x}^{*}, \alpha^{*})$.
Write $\Pset=\Pset_{\text{\textsc{Adv}}}\cap \Xset$,
 where $\Pset_{\text{\textsc{Adv}}}$ is a convex set
 corresponding to constraints~(\ref{gpmmc2}).
 Clearly, checking if $(\pmb{x}^{*}, \alpha^{*})\in \Xset$ or forming a separating hyperplane, 
  if $(\pmb{x}^{*}, \alpha^{*})\not\in \Xset$, that  boils down to
 detecting a violated constraint by $(\pmb{x}^{*}, \alpha^{*})$, 
 can be trivially  done in polynomial time, since $\Xset$ is explicitly given by a polynomial number of linear constraints.
 While either deciding that $(\pmb{x}^{*}, \alpha^{*})\in \Pset_{\text{\textsc{Adv}}}$ or
forming a separating  hyperplane 
 relies on solving the \textsc{Adv} problem  for a $\pmb{x}^{*}\in \Xset$.
Indeed, if $F(\pmb{x}^{*},\pmb{D}^{*})=\max_{\pmb{D}\in \mathcal{U}^c}F(\pmb{x}^*,\pmb{D})
\leq \alpha^{*}$ then $(\pmb{x}^{*}, \alpha^{*})\in \Pset_{\text{\textsc{Adv}}}$. Otherwise,
 a separating hyperplane is of the form:
 $\sum_{t\in [T]}F_t(\pmb{x},\pmb{D}^{*})-\alpha^{*}=0$, where
 $F_t(\pmb{x},\pmb{D}^{*}) = f_I(X_t,D^{*}_t)$ if
 $f_I(X^{*}_t,D^{*}_t) > f_B(X^{*}_t,D^{*}_t)$ and $F_t(\pmb{x},\pmb{D}^{*}) = f_B(X_t,D^{*}_t)$, 
 otherwise.
 The overall running time of the algorithm for solving (\ref{gpmmc1})-(\ref{gpmmc3}) depends on
 the running time of an algorithm for the \textsc{Adv} problem applied, since
 the ellipsoid algorithm performs a polynomial number of operations and calls to our separation oracle.
 On account of  the above remark and by Theorem~\ref{tcadv}, we get the following result:
\begin{thm}
Suppose that $\Gamma^c,\Delta_t\in \Zset_{+}$, $t\in [T]$. Then
the \textsc{MinMax} problem  under $\mathcal{U}^c$, for the non-overlapping case, 
 can be solved in a pseudopolynomial time.
 \label{tmmuc}
 \end{thm}
 Accordingly, for all  the  polynomial solvable cases of
the \textsc{Adv} problem, aforementioned in this section, one can obtain polynomial algorithms for the \textsc{MinMax} problem.
An alternative approach to solve the \textsc{MinMax} problem is a  linear programming formulation with pseudopolynomial number of constraints  and variables, which is
a constructive proof of Theorem~\ref{tmmuc}.
 Assuming that 
 $\Gamma^c,\Delta_t\in \Zset_{+}$, $t\in [T]$, we can reduce  
   problem~(\ref{auc1})-(\ref{auc3}),
that corresponds to~\textsc{Adv}, to  finding
a longest path in layered weighted graph~$G=(V,A)$ (see  Proposition~\ref{padvlp}). We can use the same reasoning as in the previous section and build the following linear programming problem for \textsc{MinMax}:
 \begin{align}
 \min\; & \pi_{\mathfrak{t}}&  \label{adlp1}\\ 
\text{s.t. }&\pi^{\delta+\delta_t}_{t}-\pi_{(t-1)}^{\delta}\geq f_I(X_t,\widehat{D}_t-\delta_t),& ((t-1)^{\delta},t^{\delta+\delta_t})\in A ,\label{adlp2}\\ 
		&\pi^{\delta+\delta_t}_{t}-\pi_{(t-1)}^{\delta}\geq f_B(X_t,\widehat{D}_t+\delta_t),& ((t-1)^{\delta},t^{\delta+\delta_t})\in A , \label{adlp2a}\\             
               &\pi_{\mathfrak{t}}-\pi_{T}^{\delta}\geq  0,& (T^{\delta},\mathfrak{t})\in A,  \label{adlp6}\\ 
                &\pi_{0}^0=0, & \label{adlp8}
\end{align}
where $\pi^{\delta}_t$ is unrestricted variable associated with node $t^\delta$  and $\pi_{\mathfrak{t}}$ is unrestricted variable associated with node $\mathfrak{t}$ of $G$.
The number of constraints and variables  in (\ref{adlp1})-(\ref{adlp8}) is $O(T\Gamma^c\Delta_{\max})$.
Now adding linear constraints 
$X_t=\sum_{i\in [t]} x_i$, $t\in [T]$, and  $\pmb{x}\in\Xset$ to (\ref{adlp1})-(\ref{adlp8}) 
 gives a linear program for 
the \textsc{MinMax} problem with  a pseudopolynomial number of constraints and variables.
It is worth pointing out that all  the  polynomially solvable cases of
the \textsc{Adv} problem, presented in this section,  can be modeled by linear programs
with polynomial numbers of constraints and variables and thus they
 apply to  \textsc{MinMax} one as well.

We now show that there exists an FPTAS for  the \textsc{MinMax} problem.
It turns out that
the formulation (\ref{gpmmc1})-(\ref{gpmmc3})
 admits  an FPTAS if there exits an FPTAS for $\max_{\pmb{D} \in \mathcal{U}^c} F(\pmb{x},\pmb{D})$
  for a given $\pmb{x}\in \Xset$
  (the \textsc{Adv} problem).  This result can easily be  adapted from~\cite[Lemma 3.5]{ASNP16}. 
  Corollary~\ref{cfadv} now implies:
 \begin{cor}
Suppose that $\Gamma^c,\Delta_t\in \Zset_{+}$, $t\in [T]$. Then
the \textsc{MinMax} problem  under $\mathcal{U}^c$, for the non-overlapping case, admits 
 an FPTAS.
 \label{cfadv1}
\end{cor}

\subsection{General case}

We now  drop the assumption 
$\widehat{D}_t+\Delta_t\leq \widehat{D}_{t+1}-\Delta_{t+1}$.
Theorem~\ref{thmcompladv} and Corollary~\ref{cmmor} now implie the following hardness results for both problems under consideration.
\begin{cor}
The \textsc{Adv} and \textsc{MinMax}  problems   for $\mathcal{U}^c$ 
in the general case
are weakly NP-hard.
\end{cor}
From an algorithmic point of view, the situation for the general case is more difficult  Recall, that for the non-overlapping case the model for the \textsc{Adv} problem has the integrality property (see Lemma~\ref{ipcb}), which allowed us to build pseudopolynomial algorithms. Now, in order to ensure that $\pmb{D}\in\mathcal{U}^c$, we have to add additional constraints $D_t\leq D_{t+1}$, $t\in [T-1]$ and the resulting model for the \textsc{Adv} problem takes the following form:
\begin{align}
 \max &\sum_{t\in [T]}\max\{ f_I(X_t, \widehat{D}_t+\delta_t), f_B(X_t, \widehat{D}_t+\delta_t)\}
 \label{gauc1}\\ 
\text{\rm s.t. }& \sum_{t\in [T]} |\delta_t|\leq \Gamma^c, & \label{gauc2}\\
                 &\widehat{D}_t+\delta_t\leq \widehat{D}_{t+1}+\delta_{t+1},& t\in [T-1],\label{gauc3}\\
                &- \Delta_t\leq \delta_t\leq \Delta_t, &t\in [T].\label{gauc4}
\end{align}
 The constraints   (\ref{gauc2})-(\ref{gauc4}) ensure that the
cumulative demand scenario~$\pmb{D}$, induced by $\pmb{\delta}$, belongs to~$\mathcal{U}^c$.
These  constraints determine 
a convex polytope. Since the objective function is convex,
 its maximum value  is attained at a vertex of this convex  polytope. However, we show an instance of the \textsc{Adv} problem which has no optimal solutions being integer, when $\Gamma^c,\Delta_t, \widehat{D}_t\in \Zset_{+}$, $t\in [T]$.
Let $T=3$, $x_1=x_2=0$, $x_3=5$, $c^B=2$, $c^I=1$, $c^P=b^P=0$, $\widehat{D}_1=3$,
$\widehat{D}_2=4$, $\widehat{D}_3=5$, $\Delta_1=3$, $\Delta_2=2$, $\Delta_3=1$ and
$\Gamma^c=4$. Now  the model  (\ref{gauc1})-(\ref{gauc4}) has the following form:
\begin{align}
\max\{14+2(\delta_1+\delta_2)-\delta_3\,:\,
&\delta_1-\delta_2\leq 1, \delta_2-\delta_3\leq 1,
|\delta_1|+|\delta_2|+|\delta_3|\leq 4, \nonumber \\
&\delta_1\in[-3,3],\delta_2\in[-2,2],\delta_3\in[-1,1] \label{exadv}
\}.
\end{align}
An easy computation  shows that $\delta^{*}_1=2\frac{1}{3}$,  $\delta^{*}_2=1\frac{1}{3}$ and 
 $\delta^{*}_3=\frac{1}{3}$ is an optimal vertex solution to (\ref{exadv}) and there is no
 optimal integer solution.
 
One can easily construct a mixed integer programming (MIP) counterpart of~(\ref{gauc1})-(\ref{gauc4})
 for the \textsc{Adv} problem,
 by linearizing (\ref{gauc1}) and (\ref{gauc2}), namely
 \begin{align}
 \max &\sum_{t\in [T]} \pi_t&\label{mgauc1}\\
  \text{\rm s.t. }&\pi_t \leq f_I(X_t, \widehat{D}_t+\delta_t) + M_t\gamma_t,         &t\in [T], \label{mgauc2}\\
                    &\pi_t \leq f_B(X_t, \widehat{D}_t+\delta_t)  + (1-M_t)\gamma_t ,  &t\in [T], \label{mgauc3}\\
                   & \sum_{t\in [T]} \beta_t \leq \Gamma^c, & \label{mgauc4}\\
                  &\delta_t\leq \beta_t, &t\in [T], \label{mgauc5}\\
                  &-\delta_t\leq \beta_t, &t\in [T],  \label{mgauc6}\\
                 &\widehat{D}_t+\delta_t\leq \widehat{D}_{t+1}+\delta_{t+1},& t\in [T-1], \label{mgauc7}\\
                &- \Delta_t\leq \delta_t\leq \Delta_t, &t\in [T], \label{mgauc8}\\
                &\beta_t \geq 0, \gamma_t\in \{0,1\},&t\in [T], \label{mgauc9}
\end{align}
 where $M_t$, $t\in [T]$, are suitably chosen large numbers.
 Unfortunately (\ref{mgauc1})-(\ref{mgauc9}) cannot be extended to a compact MIP for the \textsc{MinMax}  problem
  by using dualization.
 
 In order to cope with the \textsc{MinMax}  problem we construct a \emph{decomposition algorithm} that can be seen
 as a version of Benders' decomposition (similar algorithms 
 have been previously used in~\cite{ASNP16,AAAA17,BO08,GKZ12,SAP20,ZZ13}). The idea consists in solving a certain restricted \textsc{MinMax} problem iteratively, which provides exact or approximate solution. At each iteration an approximate production plan is computed. It is then evaluated by solving the \textsc{Adv} problem and the lower and upper bounds on the cost of an optimal production plan  for the original  \textsc{MinMax}  problem are improved.
 Consider the following linear programming program, called the master problem with $\Uset\subseteq \mathcal{U}^c$:
\begin{align}
\min \;&\alpha&\label{pmmc1}\\
\text{\rm s.t. }&\sum_{t\in [T]}\pi^{\pmb{D}}_t\leq \alpha,& \label{pmmc2}\\
&  f_I(X_t,D_t) \leq \pi^{\pmb{D}}_t, & t\in [T], \pmb{D} \in \Uset,  \label{pmmc3}\\
&   f_B(X_t,D_t) \leq \pi^{\pmb{D}}_t, & t\in [T],  \pmb{D} \in \Uset,  \label{pmmc4}\\
&   X_t=\sum_{i\in [t]} x_i, & t\in [T], \label{pmmc5}\\
&\pmb{x}\in \Xset. \label{pmmc6}
\end{align}
The constraints (\ref{pmmc2})-(\ref{pmmc4})  are the linearization of~(\ref{gpmmc2}).
Thus  (\ref{pmmc1})-(\ref{pmmc6}) is a relaxation of the \textsc{MinMax}  problem (for $\Uset=\mathcal{U}^c$
we get the \textsc{MinMax} one). An optimal solution~$\pmb{x}^{*}$ to  (\ref{pmmc1})-(\ref{pmmc6}) is an approximate
plan for \textsc{MinMax} and its quality is evaluated by solving the MIP model (\ref{mgauc1})-(\ref{mgauc9}).
In this way we get lower an upper bound on the optimal cost.
The formal description of the above decomposition procedure is presented in the form of Algorithm~\ref{aorp}.
\begin{algorithm}
\KwStep~0. $LB:=-\infty$, $UB:=+\infty$, $\Uset:=\{\widehat{\pmb{D}}\}$.\\
\KwStep~1. Solve the master problem (\ref{pmmc1})-(\ref{pmmc6}) with~$\Uset$ and derive an optimal solution
                   $(\pmb{x}^{*}, \alpha^{*})$ and update  $LB:=\alpha^{*}$. \\ 
\KwStep~2. Solve the \textsc{Adv} problem~(\ref{mgauc1})-(\ref{mgauc9}) for~$\pmb{x}^{*}$ and
                    derive an optimal (a worst-case) cumulative demand scenario~$\widehat{\pmb{D}}+\pmb{\delta}^{*}$ and 
                    update $UB:=\min\{UB, \sum_{t\in [T]}\max\{ f_I(X^{*}_t, \widehat{D}_t+\delta^{*}_t), 
                    f_B(X^{*}_t, \widehat{D}_t+\delta^{*}_t)\}\}$.\\
                          
\KwStep~3. If  $UB-LB\leq \epsilon$ then output~$\pmb{x}^{*}$ and \textbf{EXIT}.\\	
\KwStep~4. Update $\Uset:=\Uset \cup\{ \widehat{\pmb{D}}+\pmb{\delta}^{*}\}$
                    and go to Step~1.             
  \caption{A decomposition algorithm for   the \textsc{MinMax}  problem.}
  \label{aorp}
\end{algorithm}

Without loss  of generality we can assume that  $\mathcal{U}^c$  is a finite set containing only vertex scenarios,
since  the \textsc{Adv} problem attains its optimum at vertex scenarios.
This assumption ensures the convergence of  Algorithm~\ref{aorp} in  a finite number iterations.
More precisely, Algorithm~\ref{aorp} terminates after at most $O(|\mathcal{U}^c|)$ iterations~\cite[Proposition 2]{ZZ13}.
However, in practice, the decomposition based algorithms perform quite small number iterations
(see, e.g.,~\cite{AAAA17,BO08,GKZ12,SAP20,ZZ13}).

\section{Conclusions}
\label{con}
In this paper we have discussed 
 a capacitated  production planning  under uncertainty.
More specifically, we have studied   a version of the capacitated single-item lot sizing problem with backordering
under the budgeted cumulative demand uncertainty.  
We have considered two variants of the interval budgeted uncertainty representation
and used the minmax criterion to choose a best robust production plan.
For both variants, we have examined the problem of evaluating a given production plan in terms of its worst-case cost
(the \textsc{Adv} problem) and the problem of finding a robust production plan  along with its worst-case cost 
(the \textsc{MinMax}  one).
Under the discrete budgeted uncertainty, we have provided polynomial algorithms for 
the \textsc{Adv} problem  and polynomial linear programming based methods for  the \textsc{MinMax}  problem
in the non-overlapping case as well as in the general case. We have 
shown in this way that introducing uncertainty under the discrete budgeted model
does not make the problems
much computationally harder than their deterministic counterparts.
Under the continuous budgeted uncertainty the problems under consideration have different
properties than under the discrete budgeted one.
In particular, the \textsc{Adv} problem and in consequence the \textsc{MinMax}  one
have turned to be weakly NP-hard even in the non-overlapping case.
For the non-overlapping case we have constructed pseudopolynomial algorithms for the \textsc{Adv} problem
and proposed a pseudopolynomial ellipsoidal algorithm and a linear programming program
with  a pseudopolynomial number of constraints and variables for the \textsc{MinMax}  problem.
Furthermore, we have shown that both problems admit an FPTAS.
In the general case the problems still remain weakly NP-hard and, unfortunately,
there
is no easy characterization of vertex cumulative demand scenarios, namely
the integral property does not hold. We recall that this property has allowed us to build the pseudopolynomial methods
in the non-overlapping case. Accordingly, we have proposed a MIP model for the \textsc{Adv} problem and a~constraint generation algorithm for the \textsc{MinMax} problem in the general case.

There is still a number of open questions concerning the examined problems, in particular,
under the continuous budgeted uncertainty in the general case.
The  \textsc{Adv}  and \textsc{MinMax} problems are weakly NP-hard.
Thus  a full-fledged complexity 
analysis of the problems has to be carried out, i.e.  it is interesting to
 check  the existence
 pseudopolynomial algorithms, FPTASs or approximation algorithms for them.
 Furthermore, proposing a compact MIP model for the  \textsc{MinMax} problem is an interesting open problem.

\subsubsection*{Acknowledgements}
Romain Guillaume was  partially supported by the project caasc ANR-18-CE10-0012 of the French National Agency for Research.
Adam Kasperski and Pawe{\l} Zieli{\'n}ski were  supported by the National Science Centre, Poland, grant 2017/25/B/ST6/00486.

\end{document}